\documentclass[11pt,journal,draftcls,onecolumn]{IEEEtran}

\usepackage[numbers]{natbib}
\usepackage{bbm}
\usepackage{amsthm}
\usepackage{mathtools}
\usepackage{amsmath}
\usepackage{hyperref}
\usepackage{theoremref}
\usepackage{enumitem}
\usepackage{tikz}
\usetikzlibrary{shapes,arrows,calc}
\usepackage[center]{caption}
\usepackage{cancel}

\newcommand{\osum}[2]{\smashoperator[r]{\sum_{#1}} {#2}}

\newtheorem{theorem}{Theorem}
\newtheorem{lemma}{Lemma}
\newtheorem{definition}{Definition}
\newtheorem{corollary}{Corollary}
\newtheorem*{remark}{Remark}

\DeclarePairedDelimiterX{\infdivx}[2]{(}{)}{%
	#1\;\delimsize\|\;#2%
}
\DeclarePairedDelimiterX{\condx}[2]{(}{)}{%
	#1\;\delimsize|\;#2 %
}
\newcommand{\infdiv}{D\infdivx}
\newcommand{\cond}{H\condx}
\DeclarePairedDelimiter\abs{\lvert}{\rvert}%
\DeclarePairedDelimiter{\norm}{\lVert}{\rVert}
\DeclarePairedDelimiter{\para}{(}{)}

\makeatletter
\let\oldabs\abs
\def\abs{\@ifstar{\oldabs}{\oldabs*}}
\let\oldnorm\norm
\def\norm{\@ifstar{\oldnorm}{\oldnorm*}}
\makeatother

\makeatletter
\newcommand*{\rom}[1]{\expandafter\@slowromancap\romannumeral #1@}
\makeatother

\makeatletter
\newcommand{\proofpart}[2]{%
	\par
	\addvspace{\medskipamount}%
	\noindent\emph{Step #1: #2}\par\nobreak
	\addvspace{\smallskipamount}%
	\@afterheading
}
\makeatother

\newcommand{\E}{\mathrm{E}}

\newcommand{\Var}{\mathrm{Var}}

\newcommand{\Cov}{\mathrm{Cov}}

\title{Privacy-aware Distributed Hypothesis Testing in Gray-Wyner Network with Side Information}

\author{\IEEEauthorblockN{Reza Abbasalipour\IEEEauthorrefmark{1},
		Mahtab Mirmohseni\IEEEauthorrefmark{2}}
	\\
	\IEEEauthorblockA{Electrical Engineering Department, Sharif university of Technology\\
		Email: \IEEEauthorrefmark{1}reza.abbasalipour@ee.sharif.edu,
		\IEEEauthorrefmark{2}mirmohseni@sharif.edu}
	}

\begin{document}
	
	\maketitle
	
	\begin{abstract}
	The problem of distributed binary hypothesis testing in the Gray-Wyner network with side information is studied in this paper. An observer has access to a discrete memoryless and stationary source and describes its observation to two detectors via one common and two private channels. The channels are considered error-free but rate-limited. Each detector also has access to its own discrete memoryless and stationary source, i.e., the side information. The goal is to perform two distinct binary hypothesis testings on the joint distribution of observations at detectors. Additionally, the observer aims to keep a correlated latent source private against the detectors. Equivocation is used as the measure of the privacy preserved for the latent source. An achievable inner bound is derived for the general case by introducing a non-asymptotic account of the output statistics of the random binning. 
	\end{abstract}
	
	\section{Introduction}
	The problem of distributed hypothesis testing (HT) in the presence of privacy considerations for the Gray-Wyner network with side-information is investigated in this paper. The model consists of three nodes, one known as the observer and the other two known as detectors, where each has access to a separate discrete memory-less source. The observer describes its own observation to the two detectors via a network comprised of one common and two private noiseless and rate-limited channels, namely the Gray-Wyner network. Each detector, who also has access to local side information, then performs a unique simple hypothesis testing on the joint distribution of their own observation and those of the observer based on the description they have received through the channels. \\
	The observer is also interested in maintaining a level of privacy against the detectors for some latent memory-less sources correlated with the observations. These goals, performing effective hypothesis testing and maintaining privacy, seem to be contradictory and thus form a natural trade-off. If the observer provides no description to the detectors, this purpose of privacy is achieved completely. Yet, the detectors cannot perform distributed hypothesis testing based on the observation of the observer. On the other hand, if the observer can provide a perfect description, i.e., the observation itself, the result is a local hypothesis testing with an optimal solution, but the intended privacy is not preserved. In this paper, we characterize this fundamental trade-off between the communication rate of the channels, the performance achieved for the hypothesis testing, and the privacy of the observer's data. \\ 
	Our approach in addressing the hypothesis testing follows that of the Chernoff-Stein regime\cite[Section 11.8]{Cover2009CommunicationIN}. We introduce a feasible scheme and characterize its errors regarding the HT problem. The first type error is shown to be vanishing, and then the best achievable error exponent for the second type error is calculated. The goal is to acquire an error exponent for the second type error by suggesting an achievable scheme, and optimality results have been remained to be discussed in future works.   \\
	For that very purpose, we first provide a modified version of the output statistics of random binning (OSRB) framework introduced in \cite{yas14} to be used in our proposed method. Using this framework, we craft a dual problem corresponding to the original problem of distributed hypothesis testing for our network. Subsequently, the error probabilities are derived for the dual problem, which is blessed with well-defined probabilistic characteristics, almost effortlessly. Then by exhibiting the proximity of the distribution of dual setting to that of the original problem, the desired results are obtained. \\
	An advantage of such an approach is that it inherently utilizes a stochastic encoder that preserves the sources' privacy to some extent; therefore, there is no need for an additional randomizer block to deal with privacy concerns. We examine the obtained privacy in terms of equivocation measures. To our knowledge, the first use of a stochastic encoder to preserve privacy in distributed hypothesis testing was in \cite{sreekumar20} which used a likelihood encoder introduced by \cite{Song16} to maintain privacy in the Wyner-Ziv network. Prior to that, most attempts were involved adding a block to the encoder to provide an adequate obfuscation of the source observation against the detector. \\
	\subsection{Background}
	The hypothesis testing in statistics and information theory were seemingly two separate problems traditionally until recently, where many studies introduced new approaches in which they probed into statistic inference problems such as hypothesis testing using an information theory framework. Suppose one is trying to observe the data traffic in two different links and decide whether or not their traffic coincide. In the classic statistics, It is only natural that the decision making, a binary hypothesis testing in this case, needs information from both links. This means one has to send the entire traffic from at least one link to a single point for the decision to be made, a costly trivial scheme. The question that arises is that are there any other schemes that achieve the same accurate response, without having to communicate a description of the order of the data? Communication resource is a new bottleneck in this problem, coined as distributed hypothesis testing. \\
	A unified version of this problem was formulated and studied in \cite{Ahl86} where the communication bottleneck postulated as an error-free and rate-limited channel in a network similar to that of Wyner-Ziv with the addition of the side information at the detector. Although \cite{Ahl86} introduced an optimum multi-letter description of the problem, the single letter results were confined to inequalities. \cite{Han87} and \cite{sha94} improved upon these results and proved tighter bounds. \cite{rah12} devised a novel approach, built on the previous results, and proved that binning schemes yield optimum single-letter descriptions for some special cases of distributed hypothesis testing.
	Two significant expansions of this problem are the generalization of the distributed hypothesis testing to more complex networks and the introduction of the concept of privacy to the Wyner-Ziv network with side information. Among them are \cite{esc18, esc19, saleh18}, which analyzes setups with more than two entities. The concept of privacy of one legitimate entity's data against other legitimate parties is introduced in \cite{gilani19} and \cite{sreekumar20} and partly characterized. Also, \cite{liao16, liao17, liao18} investigated different privacy settings in a setup where the communication constraints are lifted.

	\subsection{Main contributions}
	This paper considers both above expansions in a single setup. To the best of our knowledge, privacy concerns have not been studied before in networks with more than two entities. One reason might be that the mathematical complexity of private distributed hypothesis testing, which is already conspicuous in the simple Wyner-Ziv network, tends to grow exponentially when more complex setups are considered. We propose a novel method based on the duality to manage the complex nature of the problem.
	\begin{enumerate}
		\item
		We introduce an approach to deal with distributed hypothesis testing problems based on the concept of duality in binning schemes \cite{yas14}.
		\item
		We establish a non-asymptotic account of output statistics of random binning and prove an achievable rate of decay. The results, which are to be used in our method, concur with \cite{yas14} in the asymptotic regime. 
		\item
		We characterize an inner bound for the general case of distributed hypothesis testing in the Gray-Wyner network with side information in the presence of privacy considerations. 
	\end{enumerate} 
	~\\
	The rest of the paper is as follows. In Section \ref{problemf}, notations and definitions to be used in this paper as well as an extensive description of the system model is introduced. In Section \ref{main}, the main results achieved in this paper are stated and then, in Section \ref{proof}, our method of choice and proof to the main results are investigated.\\
	
	\section{Preliminaries and System Model}\label{problemf}
	\subsection{Notations an Definitions}
	Here, we provide some basic notations as well as some definitions to be used in the sequel.
	We only consider discrete random variables with finite support sets. Random variables are referred to by capital letters, e.g., $X, Y$, their realization by lower case letters, e.g., $x, y$, and their support set by Calligraphic letters, e.g., $\mathcal{X}, \mathcal{Y}$. 
	A sequence of random variables $(X_i,...,X_j)$ is denoted by $X^{j}_{i}$ and its realization by $x^{j}_{i}$. In case when $i=1$ we use an abbreviated form $X^{j}$ and its corresponding realization $x^{j}$ for $(X_1,...,X_j)$. 
	Also we use $X_{\mathcal{S}}$ to denote $\{X_{j}: j \in \mathcal{S}\}$.
	The probability distribution of random variables $X$ and $Y$ is depicted as $p_{X,Y}$, their marginal distributions are denoted by $p_X$ and $p_Y$, and we use $p_{Y|X}$ to show the conditional probability distribution. 
	Sometimes we omit the argument from the notation of random variables when they match the subscription, e.g., $p_{Y|X}(y|x)=p_{Y|X}$, to keep the notation simple.
	The probability simplex of random variables $X$ and $Y$ is manifested by $\mathcal{P}(\mathcal{X} \times \mathcal{Y})$. \\
	We use $p^{U}_{\mathcal{S}}$ to refer to a uniform distribution over $\mathcal{S}$. Also $p(x^{n})$ is used for product distribution, i.e., $\prod_{i=1}^{n}p(x_i)$, unless otherwise stated. The $\mathbbm{1}(\cdot)$ refers to the indicator function. We use $H(X)$ and $\cond{X}{Y}$ to show the entropy and the conditional entropy, respectively, when the distribution of the $(X,Y)$ is clear from the context. Otherwise, we add a subscription to the notation to clarify the distribution of the random variables, e.g., $H_{p_X}(X)$ and $H_{p_{X,Y}}\condx{X}{Y}$ indicate that $(X,Y)$ is distributed according to $p_{X,Y}$ with $p_{X}$ as the marginal distribution.
		We also take advantage of the concept of random probability mass function (pmf) for discrete random variables. Random pmf of a random variable $X$ is denoted by capital letter $P_X$, so one can distinguish between pmfs and random pmfs. $P_X$ is a probability distribution over $\mathcal{P}(\mathcal{X})$.

	We first present some useful definitions. 
	\begin{definition}[Total variation distance]
		Assume $p_X$ and $q_X$ are two probability distributions on $\mathcal{X}$. The total variation distance between $p_X$ and $q_X$ is,
		\begin{equation}
			\norm{p_X - q_X}_{TV} \coloneqq \frac{1}{2}\sum_{x \in \mathcal{X}} \abs{p_X (x)- q_X(x)}.
		\end{equation}
	\end{definition}

	\begin{definition}
		For two probability mass functions $p_X$ and $q_X$ on $\mathcal{X}$, we say 			that $p_X \stackrel{\delta}{\approx} q_X$ if
		\begin{equation}
			\abs{p_X(a)-q_X(a)} < \delta \quad \text{for every} \quad a \in \mathcal{X}.
		\end{equation}
	\end{definition}

	\begin{definition}[$n$-Type]
		For any positive integer $n$, a probability mass function $p_{\bar{X}} \in \mathcal{P}(\mathcal{X})$ is referred to as an $n$-Type if for every $a \in \mathcal{X}$
		\begin{equation}
			p_{\bar{X}}(a) \in \left\{0, \frac{1}{n} , \frac{2}{n}, ..., 1\right\},
		\end{equation}
	and the set of all such $n$-types is denoted by $\mathcal{P}_{n}(\mathcal{X}) \subset \mathcal{P}(\mathcal{X})$.
	\end{definition}
	
	\begin{definition}[Type of a Sequence]
		For any positive integer $n$, the type of a sequence $x^{n} \in \mathcal{X}^{n}$ is an $n$-Type $p_{\bar{X}} \in \mathcal{P}_{n}(\mathcal{X})$, satisfying
		\begin{equation}
			p_{\bar{X}}(a) \coloneqq \frac{1}{n}\sum_{i=1}^{n}\mathbbm{1}(x_{i}=a) \quad \text{for every} \quad a \in \mathcal{X}.
		\end{equation}
	\end{definition}

	\begin{remark}
			If $x^{n}$ is a sample of $n$ observations, the type of $x^{n}$ is also called the empirical distribution of the sample $x^{n}$.
	\end{remark}
	\begin{remark}
		 The  joint type of a pair of sequences $x^{n} \in \mathcal{X}^{n}$ and $y^{n} \in \mathcal{Y}^{n}$ is defined to be the type of $\{(x_i, y_i)\}^{n}_{i=1} \in \mathcal{X}^{n} \times \mathcal{Y}^{n}$.
	\end{remark}
	\begin{remark}
		Since we make use of $n$-Types frequently in this paper, we reserve the bar notation for $n$-types to avoid any ambiguity. For example, $\bar{X} \sim p_{\bar{X}}$ depicts a random variable with the characteristics that $p_{\bar{X}} \in \mathcal{P}_{n}(\mathcal{X})$.
	\end{remark}

	\begin{definition}[Type Class]
		Having fixed an $n$-Type $p_{\bar{X}} \in \mathcal{P}_{n}(\mathcal{X})$, the set of all sequences $x^{n} \in \mathcal{X}^{n}$ whose type is $p_{\bar{X}}$ is called the type class of $p_{\bar{X}}$ and is denoted by $\mathcal{T}^{n}_{p_{\bar{X}}} \subset \mathcal{X}^{n}$.
	\end{definition}

	It's also possible to render a joint type of $\{(x_i, y_i)\}^{n}_{i=1} \in \mathcal{X}^{n} \times \mathcal{Y}^{n}$ by the type of $x^{n}$ and a stochastic matrix $p_{\bar{Y}|\bar{X}}: \mathcal{Y} \rightarrow \mathcal{X}$. The set of all such stochastic matrices is denoted by $\mathcal{P}(\mathcal{Y}|\mathcal{X})$.
	\begin{definition}[Conditional Type]
		Given $x^{n} \in \mathcal{T}^{n}_{p_{\bar{X}}}$, we say that a stochastic matrix $p_{\bar{Y}|\bar{X}}: \mathcal{Y} \rightarrow \mathcal{X} \in \mathcal{P}(\mathcal{Y}|\mathcal{X})$ is the conditional type of $y^{n} \in \mathcal{Y}^{n}$ if for every $(a, b) \in \mathcal{X} \times \mathcal{Y}$
		\begin{equation}
			p_{\bar{X},\bar{Y}}(a,b) = p_{\bar{Y}|\bar{X}}(b|a)p_{\bar{X}}(a),
		\end{equation} 
	where $p_{\bar{X},\bar{Y}}(a,b)$ is the joint type of $(x^{n},y^{n})$. The set of all conditional types, given $x^{n} \in \mathcal{T}^{n}_{p_{\bar{X}}}$, is denoted by $\mathcal{P}_{n}(\mathcal{Y}|p_{\bar{X}})$.
	\end{definition}
	
	\begin{remark}
		Given $x^{n} \in \mathcal{T}^{n}_{p_{\bar{X}}}$, the set of all conditional types, $\mathcal{P}_{n}(\mathcal{Y}|p_{\bar{X}})$, depends on $x^{n}$ only through its type. Thus, $x^{n}$ is omitted from the notation of $\mathcal{P}_{n}(\mathcal{Y}|p_{\bar{X}})$. 
	\end{remark}

	\begin{definition}[Conditional Type Class]
		Given a conditional type $p_{\bar{Y}|\bar{X}} \in \mathcal{P}_{n}(\mathcal{Y}|p_{\bar{X}})$, the set of all sequences $y^{n} \in \mathcal{Y}^{n}$ whose conditional type, given $x^{n} \in \mathcal{T}^{n}_{p_{\bar{X}}}$, is $p_{\bar{Y}|\bar{X}}$ is called the conditional type class of $p_{\bar{Y}|\bar{X}}$ and is depicted by $\mathcal{T}^{n}_{p_{\bar{Y}|\bar{X}}}(x^{n})$.
	\end{definition}
	
	\begin{remark}
		The size of a conditional type class, namely $\mathcal{T}^{n}_{p_{\bar{Y}|\bar{X}}}(x^{n})$, depends on $x^{n}$ only through its type.
	\end{remark}

	\begin{definition}[Constant-Composition Distribution]
		For a fixed integer $n$, suppose we are given an $n$-type $p_{\bar{X}}$. A constant-composition distribution on $\mathcal{X}^{n}$ according to the $p_{\bar{X}}$ is defined as:
		\begin{equation}
			p(x^{n}) = \frac{1}{\abs{\mathcal{T}^{n}_{p_{\bar{X}}}}}\mathbbm{1}\left\{
			x^{n} \in \mathcal{T}^{n}_{p_{\bar{X}}}
			\right\}.
		\end{equation}
	\end{definition}
	

	\subsection{System Model and Problem Formulation}
	We consider the problem of distributed hypothesis testing in the Gray-Wyner network with side information in the presence of \textit{privacy} considerations, which we refer to as the GWP problem. Assume a tuple of discrete memoryless stationary sources $(X^{n},Z^{n}_{1},Z^{n}_{2},S^{n}_{1},S^{n}_{2})$ distributed on the discrete set $\mathcal{X}^{n} \times \mathcal{Z}^{n}_{1} \times \mathcal{Z}^{n}_{2} \times \mathcal{S}^{n}_{1} \times\mathcal{S}^{n}_{2}$. The observer observes $(X^{n},S^{n}_{1},S^{n}_{2})$, the first detector has access to $Z^{n}_{1}$ and the second detector has access to $Z^{n}_{2}$. The goal is to perform a hypothesis testing on $(X^{n},Z^{n}_{1},Z^{n}_{2})$ while preserving the privacy of $(S^{n}_{1},S^{n}_{2})$ against the detectors. Upon observing $X^{n}$, the observer generates three message indices $(M_0,M_1,M_2)$ using $(M_0,M_1,M_2) = f_{n}(X^{n})$, where $f^{n}: \ \mathcal{X}^{n} \rightarrow \mathcal{M}_0 \times \mathcal{M}_1 \times \mathcal{M}_2$ and $\mathcal{M}_{i} \coloneqq [2^{nR_{i}}]$ for $i \in \{1,2,3\}$. The index $M_{j}$ for $j \in \{1,2\}$ is sent to Detector $j$ through a private channel. Alongside them, the message index $M_0$ is sent to both the detectors through a common channel. All channels are assumed to be error-free. Also note that $f^{n}$ could be a stochastic function.
	Now that the detector $j \in \{1,2\}$ has access to $(M_0,M_j,Z^n_j)$, it can take advantage of a decoding function to perform the desired hypothesis testing. Also the detectors do not have any direct access to $(S^{n}_{1},S^{n}_{2})$, but detector $j$ is interested in obtaining as much information as possible about $S^{n}_{j}$, an goal that the observer deprecates and tries to keep out of reach. 
	\\

	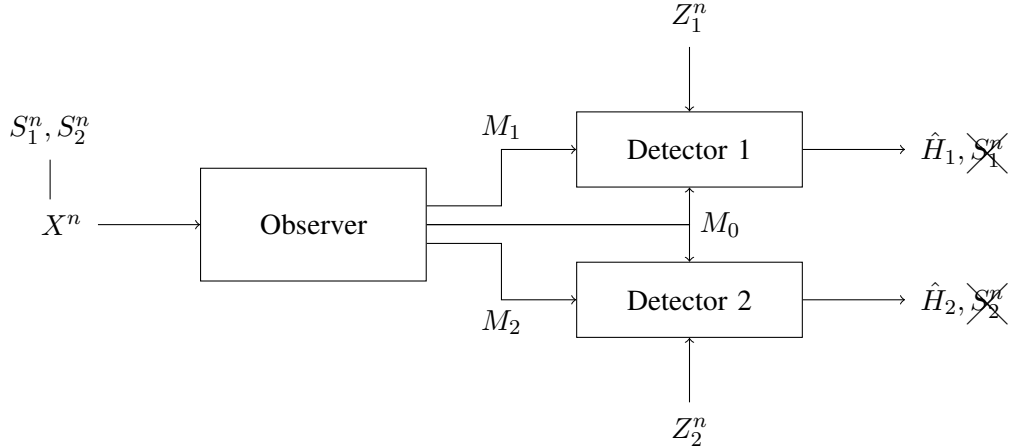
\begin{figure}[!h]
	\centering
	\begin{tikzpicture}

		\node (Tr) at (0,0) [shape=rectangle, draw, minimum width=3cm, minimum height=1.5cm] {Observer};
		
		\node (Rx1) at (5,1) [shape=rectangle, draw, minimum width=3cm, minimum height=1cm] {Detector 1};
		\node (Rx2) at (5,-1) [shape=rectangle, draw, minimum width=3cm, minimum height=1cm] {Detector 2};
		
		\draw [->] (Tr.east) ++(0,0.25) -- ($(Tr)+(2.5,0.25)$) |- node[above] {$M_1$} (Rx1.west);
		\draw [->] (Tr.east) ++(0,-0.25) -- ($(Tr)+(2.5,-0.25)$) |- node[below] {$M_2$} (Rx2.west);
		
		\draw [->] (Tr.east) -| node[right] {$M_0$} (Rx1.south);
		\draw [->] (Tr.east) -| (Rx2.north);
		
		\node (source0) at ($(Tr)+(-3,0)$) [label={[label distance=-2mm]left: {$X^{n}$}}] {};  
		\draw [->] (source0) -- (Tr.west) {};
		
		\node (latent) at ($(source0)+(-0.5,1)$) [label={[label distance=-2mm]above: {$S^{n}_{1},S^{n}_{2}$}}] {};  
		\draw [-] (latent.south) -- ($(source0.north)+(-0.5,0.2)$) {};
		
		\node (res1) at ($(Rx1)+(3,0)$) [label={[label distance=-2mm]right: {$\hat{H}_{1}, \xcancel{S^{n}_{1}}$}}] {};  
		\node (res2) at ($(Rx2)+(3,0)$) [label={[label distance=-2mm]right: {$\hat{H}_{2}, \xcancel{S^{n}_{2}}$}}] {};  
		\draw [->] (Rx1.east) -- (res1);
		\draw [->] (Rx2.east) -- (res2);
		
		\node (source1) at ($(Rx1)+(0,1.5)$) [label={[label distance=-2mm]above: {$Z^{n}_{1}$}}] {};  
		\node (source2) at ($(Rx2)+(0,-1.5)$) [label={[label distance=-2mm]below: {$Z^{n}_{2}$}}] {};  
		\draw [->] (source1.south) -- (Rx1.north);
		\draw [->] (source2.north) -- (Rx2.south);
		
	\end{tikzpicture}
	\caption{Setup of the GWP problem}
\end{figure}
	
	We are considering the binary hypothesis testing in which there are only two hypotheses. The hypothesis test is performed by each of the detectors on the joint distribution of  $(X^{n},Z^{n}_{1},Z^{n}_{2})$ where the null hypothesis is,
	\[
		H_0: \ (X^{n},Z^{n}_{1},Z^{n}_{2}) \sim  \prod_{i=1}^{n}p_{X,Z_1,Z_2},
	\]
	and the alternate hypothesis is,
	\[
	H_1: \ (X^{n},Z^{n}_{1},Z^{n}_{2}) \sim  \prod_{i=1}^{n}q_{X,Z_1,Z_2}.
	\]
	The true hypothesis random variable is denoted by $H$ and the output of the hypothesis testing by each of the detectors is depicted as $\hat{H}_{i}$ for detector $i \in \{1,2\}$. Since the first detector only observes $Z^{n}_{1}$ and a function of the $X^{n}$, it must perform the hypothesis testing on the marginal distribution on $(X^{n},Z^{n}_{1})$ by using $g^{n}_{1}: \ \mathcal{Z}^{n}_1 \times \mathcal{M}_0 \times \mathcal{M}_1 \rightarrow \{0,1\}$ as the decision rule which outputs
	\[
	\hat{H}_{1} = g^{n}_{1}(Z^{n}_{1},M_0,M_1).
	\]
	The second detector performs the same hypothesis testing on $(X^{n},Z^{n}_{2})$ using $g^{n}_{2}: \ \mathcal{Z}^{n}_2 \times \mathcal{M}_0 \times \mathcal{M}_2 \rightarrow \{0,1\}$ as  
	\[
	\hat{H}_{2} = g^{n}_{2}(Z^{n}_{2},M_0,M_2).
	\]
	The type \rom{1} and type \rom{2} type errors are defined as 
	\[
	\alpha_{n,i}(f^{n},g^{n}_{i}) \coloneqq \Pr\left(\hat{H}_{i}=1|H=0\right) \quad \text{for} \quad i \in \{1,2\},
	\]
	and
	\[
	\beta_{n,i}(f^{n},g^{n}_{i}) \coloneqq \Pr\left(\hat{H}_{i}=0|H=1\right) \quad \text{for} \quad i \in \{1,2\},
	\]
	respectively, where $i$ refers to the detector $i \in \{1,2\}$. 
	Notice that should the marginal distributions of the two hypotheses be different, the observer and the two detectors can conveniently and independently perform the hypothesis test on marginal distributions based on their local observations, yielding a vanishing type \rom{1} errors and an exponential type \rom2 errors fading to zero. In this paper, we assume that the two hypotheses distributions have a same marginal distribution i.e., $p_X=q_X$. 
	We measure the performance of a hypothesis testing scheme by measuring the achievable exponent for type \rom{2} errors i.e. $-\frac{1}{n}\log(\beta_{n,i}(f^{n},g^{n}_{i}))$, having fixed upper bounds for type \rom{1} errors.  
	Given a constraint set $(\epsilon_{n,1},\epsilon_{n,2})$ on the type \rom{1} errors, we are looking for an scheme with feasible type \rom{1} errors and the best achievable type \rom{2} error exponent pair, namely $(-\frac{1}{n}\log(\beta_{n,1}(f^{n},g^{n}_{1})),-\frac{1}{n}\log(\beta_{n,2}(f^{n},g^{n}_{2})))$.
	\\
	As we mentioned earlier, another aspect to this problem is that the first detector is curious about the latent random variable $S^{n}_{1}$ while the second detector is focused on the information it can obtain about $S^{n}_{2}$. The pair $(S^{n}_{1},S^{n}_{2})$ is constructed in an i.i.d manner whose one-shot marginal distribution $p_{S_{1},S_{2}}$ is consistent regardless of the true hypothesis. As we desire to conceal $S^{n}_{1}$ from the first detector and $S^{n}_{2}$ from the second one, we call $(S^{n}_{1},S^{n}_{2})$ the private part of the observation at the observer or simply the private data. We use \textit{equivocation} defined as $\frac{1}{n}H(S^{n}_{i}|Z^{n}_{i},M_0,M_i)$ for $i \in \{1,2\}$ for the measure of privacy. The perfect privacy is achieved if we have $H(S^{n}_{i}|Z^{n}_{i},M_0,M_i)=H(S^{n}_{i}|Z^{n}_{i})$ i.e.,
	\[
		I \condx{S^{n}_{i};M_0,M_i}{Z^{n}_{i}} = 0 \quad for \quad i \in \{1,2\}.
	\]
	The goal is to achieve the best error-exponent for the type \rom{2} error while preserving the constraints on the type \rom{1} errors and a certain level of the privacy for private data against the detectors. To attain such a goal, first we need to define achievability criteria for the problem.
	
	\begin{definition} \label{achievable}
		Assume a rate vector $\boldsymbol{R}=(R_1,R_2,R_3) \in \mathbbm{R}^{3}_{+}$, a privacy vector $\boldsymbol{\Lambda}=(\Lambda_1,\Lambda_2) \in \mathbbm{R}^{2}_{+}$, and a type \rom{2} error exponent vector $\boldsymbol{\theta}=(\theta_1,\theta_2)  \in \mathbbm{R}^{2}_{+}$.
		For a specified type \rom{1} error constraint, $\boldsymbol{\epsilon}=(\epsilon_1,\epsilon_2) \in [0, 1)^2$, the tuple $(\boldsymbol{\theta}, \boldsymbol{R},\boldsymbol{\Lambda})$ is achievable if there exists a sequence of encoder and decoder functions $(f^{n},g^{n}_{1},g^{n}_{2})$ such that,
		\begin{align}
			&\limsup_{n \rightarrow \infty} \alpha_{n,i}(f^{n},g^{n}_{i}) \leq \epsilon_i \quad \text{for} \quad i \in \{1,2\}, \\
			&\limsup_{n \rightarrow \infty} -\frac{1}{n}\log{\beta_{n,i}(f^{n},g^{n}_{i})} \geq \theta_i \quad \text{for} \quad i \in \{1,2\}, \\
			&H(S^{n}_{i}|Z^{n}_{i},M_0,M_i) \geq n\Lambda_i \quad \text{for} \quad i \in \{1,2\}.
		\end{align}
		The achievable region $\mathcal{R}(\boldsymbol{\epsilon})$ is the closure of the set of all achievable tuples $(\boldsymbol{\theta}, \boldsymbol{R},\boldsymbol{\Lambda})$, given a specific $\boldsymbol{\epsilon}$.
	\end{definition}

	In the next section, we are going to introduce an inner bound on the $\mathcal{R}(\boldsymbol{\epsilon})$.

	\section{Main result}\label{main}
	The following theorem provides the main result of this paper by devising an inner bound on $\mathcal{R}(\boldsymbol{\epsilon})$. 
	
	\begin{theorem}\thlabel{mainresult}
		Given $\boldsymbol{\epsilon}=(\epsilon_1,\epsilon_2) \in [0, 1)^2$, the $(\boldsymbol{\theta}, \boldsymbol{R},\boldsymbol{\Lambda}) \in \mathcal{R}(\boldsymbol{\epsilon})$ is achievable, if there exist auxiliary random variables $Y_{[0:2]}$ with $p_{Y_{[0:2]}|X}$ such that the following conditions hold:
		\begin{align}
			\theta_j &\leq \theta^{*}_{j}, \\
			\Lambda_i &\leq H(S_i|Z_i,Y_{0},Y_{i}),
		\end{align}
		\begin{equation}
			\begin{gathered}
				R_{0} > \max_{i \in \{1,2\}}\{ I(X;Y_{0}|Z_{i}) - I(Y_{0},Y_{i}|Z_{i}) \},\\
				R_{1} > I(X;Y_{1}|Z_{1}) - I(Y_{0},Y_{1}|Z_{1}),\\
				R_{2} > I(X;Y_{2}|Z_{2}) - I(Y_{0},Y_{2}|Z_{2}),\\
				R_{0} + R_{1} > I(X;Y_{0}Y_{1}|Z_{1}), \\
				R_{0} + R_{2} >  I(X;Y_{0}Y_{2}|Z_{2}), \\
				R_{0} + R_{1} > I(X;Y_{0}|Z_{2}) + I(X;Y_{1}|Y_{0}Z_{1}) - I(Y_{0};Y_{2}|Z_{2}),\\
				R_{0} + R_{2} > I(X;Y_{0}|Z_{1}) + I(X;Y_{2}|Y_{0}Z_{2}) - I(Y_{0};Y_{1}|Z_{1}),\\
				R_{1} + R_{2} > I(X;Y_{1}|Y_{0}Z_{1}) + I(X;Y_{2}|Y_{0}Z_{2}) + I(Y_{1};Y_{2}|XY_{0}) - I(Y_{1}Y_{2};Y_{0}|X),\\
				R_{0} + R_{1} + R_{2} >   I(X;Y_{1}|Y_{0}Z_{1}) + I(X;Y_{2}|Y_{0}Z_{2}) + \max_{i \in \{1,2\}}\{I(Y_{0};X|Z_{i})\} + I(Y_{1};Y_{2}|XY_{0}),\\
				2R_{0}+R_{1}+R_{2} >I(X;Y_{1}|Y_{0}Z_{1})\!+\!I(X;Y_{2}|Y_{0}Z_{2})\!+\!I(Y_{0};X|Z_{1})\!+\!I(Y_{0};X|Z_{2})\!+\!I(Y_{1};Y_{2}|XY_{0}),
			\end{gathered}
		\end{equation}
		for $j \in \{1,2\}$, where
		\begin{equation*}
			\begin{gathered}
				\theta^{*}_{j} = \min\left\{
				E_{0,j}(p_{Y_{[0:2]}|X}),
				E_{1,j}(p_{Y_{[0:2]}|X}),
				E_{2,j}(p_{Y_{[0:2]}|X})
				\right\},
				\\
				E_{0,j}(p_{Y_{[0:2]}|X}) \coloneqq
				\min_{\pi_{X,Y_{[0:2]},Z_j} \in \mathcal{K}_{0}} \infdiv*{\pi_{X,Y_{[0:2]},Z_j}}{q_{X,Z_j}p_{Y_{[0:2]}|X}},
				\\
				E_{1,j}(p_{Y_{[0:2]}|X})
				\coloneqq
				\min_{\pi_{X,Y_{[0:2]},Z_j} \in \mathcal{K}_{1,j}} \infdiv*{\pi_{X,Y_{[0:2]},Z_j}}{q_{X,Z_j}p_{Y_{[0:2]}|X}}
				+
				\min_{\emptyset \neq \mathcal{S} \subseteq \{0,j\}}
				\para*{\sum_{i \in \mathcal{S}}R_i + \tilde{R}_i - H(Y_{\mathcal{S}}|Z_j,Y_{\mathcal{S}^{c}})},
				\\
				E_{2,j}(p_{Y_{[0:2]}|X}) \coloneqq 
				\min_{\pi_{X,Y_{[0:2]},Z_j} \in \mathcal{K}_{2,j}}
				\left \{
				\infdiv*{\pi_{X,Y_{[0:2]},Z_j}}{q_{X,Z_j}p_{Y_{[0:2]}|X}}
				+
				\frac{1}{2}\left[
				\min_{\mathcal{S} \subseteq [0:2]}
				\left(
				H_{\pi}(Y_{\mathcal{S}}|X) -
				\sum_{i \in \mathcal{S}}\tilde{R}_{i} 
				\right)
				\right]^{+}
				\right\},
			\end{gathered}
		\end{equation*}
		and
		\begin{equation*}
			\begin{gathered}
				\mathcal{K}_{0,j}= \left\{
				\pi_{X,Y_{[0:2]},Z_j} \in \mathcal{P}\left(\mathcal{X}\times\mathcal{Y}_{[0:2]}\times\mathcal{Z}_j\right) \ : \
				\pi_{X,Y_{[0:2]}} = p_{X,Y_{[0:2]}} \wedge
				\pi_{Y_{0},Y_{j},Z_j}=p_{Y_{0},Y_{j},Z_j}
				\right\},
				\\
				\mathcal{K}_{1,j}= \left\{
				\pi_{X,Y_{[0:2]},Z_j} \in \mathcal{P}\left(\mathcal{X}\times\mathcal{Y}_{[0:2]}\times\mathcal{Z}_j\right) \ : \
				\pi_{X,Y_{[0:2]}} = p_{X,Y_{[0:2]}} \wedge
				\pi_{Z_j}=p_{Z_j}
				\right\},
				\\
				\mathcal{K}_{2,j}= \left\{
				\pi_{X,Y_{[0:2]},Z_j} \in \mathcal{P}\left(\mathcal{X}\times\mathcal{Y}_{[0:2]}\times\mathcal{Z}_j\right) \ : \
				\pi_{Z_j}=p_{Z_j}
				\right\},
			\end{gathered}
		\end{equation*}
	\begin{equation*}
		\begin{gathered} 
			\tilde{R}_{0} < H(Y_{0}|X), \\
			\tilde{R}_{1} < H(Y_{1}|X), \\
			\tilde{R}_{2} < H(Y_{2}|X), \\
			\tilde{R}_{0}  + \tilde{R}_{1} < H(Y_{0}Y_{1}|X), \\
			\tilde{R}_{0}  + \tilde{R}_{2} < H(Y_{0}Y_{2}|X), \\
			\tilde{R}_{1}  + \tilde{R}_{2} < H(Y_{1}Y_{2}|X), \\
			\tilde{R}_{0}  + \tilde{R}_{1} + \tilde{R}_{2} <  H(Y_{0}Y_{1}Y_{2}|X).
		\end{gathered}
	\end{equation*}
	\end{theorem}
	\begin{remark}
		In this paper, we only consider the problem for the Gray-Wyner network, which we call GWP. However, since the proof offers a comprehensive framework for different setups, in view of the fact that our approach doesn't concern the specific features of the Gray-Wyner network, the proof could be applied to other networks almost effortlessly. 
	\end{remark}
	
	\section{Proof of the Main Result}\label{proof}
	To prove that we can achieve the specific exponent for the type \rom{2} errors' rate of decay while maintaining a vanishing type \rom{1} errors, stated in \thref{mainresult}, we propose a scheme for the GWP setup and then evaluate the probability of its error events induced by its distribution.\\
	The scheme is comprised of an encoder and two separate decoders for each of the detectors, which will be introduced in the subsequent parts of the proof. Since privacy is another issue to consider, the encoder is a stochastic block that takes advantage of a few random binning blocks. The resulted distribution is a random pmf, meaning that we have to show the probability of errors satisfy the constraints in \thref{mainresult} in the mean and then deduce that there are fixed encoders and decoders that also are consistent with the constraints. \\
	The random pmf induced by the random mappings and the stochastic characteristics of the proposed encoder is not easy to evaluate. On the other hand, the random mappings behave smoothly in the mean with a tractable distribution which can be dealt with easily. Suppose we can show that the random pmf induced by the encoder has concentration properties. In that case, we can craft a dual setup with a distribution similar to the mean distribution of the encoder. Then we can evaluate the probability of error events in the dual problem more easily. Consequently, using the concentration properties of the encoder's random pmf, we can show that the results are also applicable to the main problem by making some adjustments. \\
	To follow this approach, first, in Subsection \ref{osrbtheorem}, we ascertain the aforementioned concentration properties of the distributed random binning, and then proceed, in Subsection \ref{ourapproach}, to complete the proof by introducing a dual problem for the GWP setup, evaluating the error events in the dual problem, and attributing the results to the GWP setup, as described. \\
	Finally, we find a lower bound on the equivocation measure of our private data by using the same method as the error exponents in Subsection \ref{privacy}. We first find a lower bound on the equivocation measure in the dual problem and then ascertain that the results are roughly applicable to the main problem.

	\subsection{Non-asymptotic output statistics of random binning}\label{osrbtheorem}
	Let $(Y_{[1:T]},X)$ be discrete memoryless stationary sources distributed according to a joint pmf $p_{Y_{[1:T]},X}$ on the discrete set $\prod_{i=1}^{T}\mathcal{Y}_{i}\times\mathcal{X}$. A distributed random binning scheme can be defined as a set of $T$ random mappings, each described by $\mathcal{B}_i : \mathcal{Y}^{n}_{i} \rightarrow [1: 2^{nR_{i}}]$ for $i \in [1:T]$, where $\mathcal{B}_i$ maps each sequence of $\mathcal{Y}^{n}_{i}$ uniformly and independently to $[1: 2^{nR_{i}}]$. We denote the random variable $\mathcal{B}_i(\cdot)$ by simply $B_i$. Also the realization of the $B_i$ will be depicted as $b_i$. \\
	The distributed random binning scheme will induce a random pmf through the inherent randomness in each of the described random binnings, namely
	\[
	P(y^{n}_{[1:T]},x^{n},b_{[1:T]})=p(y^{n}_{[1:T]},x^{n})\prod_{i=1}^{T}\mathbbm{1}(\mathcal{B}_i(y^{n}_{i})=b_i).
	\]
	The induced random pmf is called the output statistics of random binning (OSRB). The OSRB theorem in \cite{yas14} states that given a specific criteria on the binning rates, i.e., $(R_1,\ldots,R_T)$, the induced random pmf has a concentration property and its expected deviation from its mean would vanish asymptotically in terms of total variation distance.

	\begin{lemma}{\cite[Theorem 1]{yas14}}\thlabel{OSRBorigin}
		if for each $\mathcal{S} \subseteq [1:T]$ the following constraints holds
		\begin{equation}
			\sum_{i \in \mathcal{S}}R_{i} < \cond{Y_{\mathcal{S}}}{X},
		\end{equation}
		then as $n \rightarrow \infty$ we would have
		\begin{equation}
			\mathbbm{E}_{\mathcal{B}}\norm{P(x^{n},b_{[1:T]}) - \mathbbm{E}_{\mathcal{B}}P(x^{n},b_{[1:T]})}_{TV} \rightarrow 0,
		\end{equation}
		where $\mathcal{B}$ is the set of all random mappings, i.e. $\mathcal{B}=\{\mathcal{B}_{i}: i \in [1:T]\}$.
	\end{lemma}
	Since in this paper we deal with the exponential rates of decay, we need a non-asymptotic account of how distributed binning scheme behaves. The following theorem provides a non-asymptotic version of \thref{OSRBorigin}. 
	
	\begin{theorem}\thlabel{OSRBnew}
		Suppose $(Y_{[1:T]},X)$ to be discrete memoryless stationary sources with $p_{Y_{[1:T]},X}$ as the joint pmf on $\prod_{i=1}^{T}\mathcal{Y}_{i}\times\mathcal{X}$. Also assume we have a set of random binnings, each denoted by $\mathcal{B}_i : \mathcal{Y}^{n}_{i} \rightarrow [1: 2^{nR_{i}}]$ for $i \in [1:T]$, where $\mathcal{B}_i$ maps each $\mathcal{Y}^{n}_{i}$ uniformly and independently to $[1: 2^{nR_{i}}]$, then the following constraint holds
		\[
			\begin{gathered}
				-\frac{1}{n}\log\mathbbm{E}_{\mathcal{B}}\norm{P(x^{n},b_{[1:T]}) - \mathbbm{E}_{\mathcal{B}}P(x^{n},b_{[1:T]})}_{TV} \geq \\
				\min_{\pi_{Y_{[1:T]},X} \in \mathcal{P}(\mathcal{Y}_{[1:T]}\times\mathcal{X})}\left\{
				\infdiv{\pi_{Y_{[1:T]},X}}{p_{Y_{[1:T]},X}}
				+
				\frac{1}{2}\left[
				\min_{\mathcal{S} \subseteq [1:T]} \left\{
				H_{\pi}\condx{Y_{\mathcal{S}}}{X} -
				\sum_{i \in \mathcal{S}}R_{i} -
				\delta^{\mathcal{S}}_{n}
				\right\}
				\right]^{+}
				- \epsilon_{n}
				\right\},
			\end{gathered}
		\]
		where $\epsilon_{n} \coloneqq \abs{\mathcal{X}}\abs{\mathcal{Y}_{[1:T]}} 
		\frac{\log(n+1)}{n}$ and  $\delta^{\mathcal{S}}_{n} \coloneqq \abs{\mathcal{X}}\abs{\mathcal{Y}_{\mathcal{S}}} 
		\frac{\log(n+1)}{n}+
		\frac{T}{n}$ converge to zero as $n \rightarrow \infty$. $\mathcal{B}$ is the set of all random mappings, i.e., $\mathcal{B}=\{\mathcal{B}_{i}: i \in [1:T]\}$.
	\end{theorem}
	\begin{proof}
		The proof is provided in Appendix \ref{OSRBnewProof}.
	\end{proof}

	\begin{remark}
		In the case when $\sum_{i \in \mathcal{S}}R_{i} \geq \cond{Y}{X}$ for some arbitrary $\mathcal{S} \subseteq [1:T]$, the optimal choice would be  $\pi_{Y_{[1:T]},X}=p_{Y_{[1:T]},X}$, yielding the zero exponent. This observation coincides with our perception from \thref{OSRBorigin} for high-rate codes.
		
	\end{remark}
	\begin{remark}
		For convenience, let's define 
		\[
			\zeta(R_{\mathcal{T}},p_{X},p_{Y_{\mathcal{T}}|X}) \coloneqq 
			\min_{\pi_{Y_{\mathcal{T}},X} \in \mathcal{P}(\mathcal{Y}_{\mathcal{T}}\times\mathcal{X})}\left\{
			\infdiv{\pi_{Y_{\mathcal{T}},X}}{p_{Y_{\mathcal{T}},X}}
			+
			\frac{1}{2}\left[
			\min_{\mathcal{S} \subseteq \mathcal{T}} \left\{
			H_{\pi}\condx{Y_{\mathcal{S}}}{X} -
			\sum_{i \in \mathcal{S}}R_{i} -
			\delta^{\mathcal{S}}_{n}
			\right\}
			\right]^{+}
			- \epsilon_{n}
			\right\},
		\]
		where $\mathcal{T} \coloneqq [1:T]$ and $R_{\mathcal{T}} \coloneqq \{R_{i}: i \in \mathcal{T}\}$.
	\end{remark}

	Another variant of \thref{OSRBnew}, which is needed in this paper, is a case of distributed random binning when there is another discrete random sequence $Z^{n}$, correlated with $X^{n}$ in a manner that $Z^{n} \leftrightarrow X^{n} \leftrightarrow Y^{n}_{[1:T]}$ forms a Markov chain. The ensued distribution on $\prod_{i=1}^{T}\mathcal{Y}^{n}_{i}\times\mathcal{X}^{n} \times \mathcal{Z}^{n}$ can be presented as $p(y^{n}_{[1:T]},x^{n},z^{n})=p(z^{n})p(x^{n}|z^{n})p(y^{n}_{[1:T]}|x^{n})$ where $p(x^{n}|z^{n})$ and $p(y^{n}_{[1:T]}|x^{n})$ are product distributions. We also assume that the $Z^{n}$ has a constant-composition distribution on $\mathcal{Z}^{n}$ with respect to a specific $n$-Type $p_{\bar{Z}}$, i.e., 
	\begin{equation}
		p(z^{n}) = \frac{1}{\abs{\mathcal{T}^{n}_{p_{\bar{Z}}}}}\mathbbm{1}\left\{
		z^{n} \in \mathcal{T}^{n}_{p_{\bar{Z}}}
		\right\}.
	\end{equation}
	Note that the constant composition distribution, and consequently, the $p(y^{n}_{[1:T]},x^{n},z^{n})$ are not product distributions. The following theorem presents this extension.
	
	\begin{theorem}\thlabel{OSRBnewconditioned}
		Let $(Y^{n}_{[1:T]},X^{n},Z^{n})$ be discrete sources given that $Z^{n} \leftrightarrow X^{n} \leftrightarrow Y^{n}_{[1:T]}$. Assume we have $p(y^{n}_{[1:T]},x^{n},z^{n})=p(z^{n})p(x^{n}|z^{n})p(y^{n}_{[1:T]}|x^{n})$ where $p(x^{n}|z^{n})$ and $p(y^{n}_{[1:T]}|x^{n})$ are product distributions. Also assume a distributed random binning scheme comprised of $\mathcal{B}_i : \mathcal{Y}^{n}_{i} \rightarrow [1: 2^{nR_{i}}]$ for $i \in [1:T]$. The following constraint holds
		\[
			\begin{gathered}
				-\frac{1}{n}\log\mathbbm{E}_{\mathcal{B}}\norm{P(x^{n},b_{[1:T]}) - \mathbbm{E}_{\mathcal{B}}P(x^{n},b_{[1:T]})}_{TV} \geq \\
				\min_{\pi_{Y_{[1:T]},X|Z} \in \mathcal{P}(\mathcal{Y}_{[1:T]}\times\mathcal{X}|\mathcal{Z})}\left\{
				\infdiv{\pi_{Y_{\mathcal{T}},X|Z}}{p_{Y_{\mathcal{T}},X|Z}|p_{\bar{Z}}}
				+
				\frac{1}{2}\left[
				\min_{\mathcal{S} \subseteq [1:T]} \left\{
				H_{\pi}\condx{Y_{\mathcal{S}}}{X} -
				\sum_{i \in \mathcal{S}}R_{i} -
				\delta^{\mathcal{S}}_{n}
				\right\}
				\right]^{+}
				- \epsilon_{n}
				\right\},
			\end{gathered}
		\]
		where  $\epsilon_{n} \coloneqq \abs{\mathcal{X}}\abs{\mathcal{Y}_{[1:T]}} 
		\frac{\log(n+1)}{n}$ and  $\delta^{\mathcal{S}}_{n} \coloneqq \abs{\mathcal{X}}\abs{\mathcal{Y}_{\mathcal{S}}} 
		\frac{\log(n+1)}{n}+
		\frac{T}{n}$ converge to zero as $n \rightarrow \infty$. $\mathcal{B}$ is the set of all random mappings, i.e. $\mathcal{B}=\{\mathcal{B}_{i}: i \in [1:T]\}$.
	\end{theorem}
	\begin{proof}
		The proof is provided in Appendix \ref{OSRBnewconditionedProof}.
	\end{proof}
	\begin{remark}
		We use the following definition to refer to the acquired exponent:
		\[
		\aleph(R_{\mathcal{T}},p_{X^{n}},p_{Y_{\mathcal{T}}|X}) \coloneqq 		
		\min_{\pi \in \mathcal{P}(\mathcal{Y}_{\mathcal{T}}\times\mathcal{X}|\mathcal{Z})}\left\{
		\infdiv{\pi_{Y_{\mathcal{T}},X|Z}}{p_{Y_{\mathcal{T}},X|Z}|p_{\bar{Z}}}
		+
		\frac{1}{2}\left[
		\min_{\mathcal{S} \subseteq \mathcal{T}} \left\{
		H_{\pi}\condx{Y_{\mathcal{S}}}{X} -
		\sum_{i \in \mathcal{S}}R_{i} -
		\delta^{\mathcal{S}}_{n}
		\right\}
		\right]^{+}
		- \epsilon_{n}
		\right\}.
		\]
	\end{remark}
	
	\subsection{Proof of \thref{mainresult}}\label{ourapproach}
	
	Our approach in proving an achievable exponent vector for the GWP problem is comprised of few steps. In the first step, or \textit{step (1) of the proof}, we modify the main problem by adding a shared randomness to it and then fabricate a well-defined dual problem \textit{(Protocol A)} for our modified main setup \textit{(Protocol B)}.
	In the second step or \textit{step (2a) of the proof}, we solve the distributed hypothesis testing for the dual problem and determine its error bounds, and in the third step or \textit{step (2b) of the proof} we will explore the criteria in which the distributions of the modified main problem and the dual problem are almost identical, and therefore the results for \textit{Protocol A} are also applicable to \textit{Protocol B} to some extent. 
	In the last step or \textit{step (3) of the proof}, we show that we obtain the desired results for the main problem by eliminating the shared randomness from its modified version.  

	\proofpart{1}{Introducing the dual problem.}
	In this step, a modified version of the main problem, which we call \textit{Protocol B}, along with its corresponding dual problem, \textit{Protocol A}, is introduced and their induced distribution will be looked at.\\

	\textit{Protocol A (source coding side of the problem)}: 
	Define three auxiliary random variables $Y_{[0:2]}$ and fix the conditional distribution $p_{Y_{[0:2]}|X}$ such that:
\[
Y_{[0:2]} \longleftrightarrow X \longleftrightarrow Z_{[1:2]}.
\]
	 Recall that the two competing hypotheses have a same marginal distribution, namely $p_{X}$ for random variable $X$. Let $(Y^{n}_{[0:2]},X)$ be a sequence distributed according to $\prod_{t=1}^{n}p_{X}p_{Y_{[0:2]}|X}$. Now for each $i \in [0:2]$, consider a random binning where two bin indices $m_i \in [1:2^{nR_i}]$ and $f_i \in [1:2^{n\tilde{R}_i}]$ are assigned to each $y^{n}_{i}$, uniformly and independently, denoted by $\mathcal{B}_{M,i}$ and $\mathcal{B}_{F,i}$ respectively. Further, consider two distinct decoders depicted as $j \in \{1,2\}$, each trying to perform the hypothesis testing based on their observations. Decoder $j$, $j \in \{1,2\}$, has access to $(M_0,M_j,Z^{n}_{j})$ and will be manifested by its induced distribution, $P^{HT}(\hat{h}_{j}|m_0,f_0,m_j,f_j,z^{n}_{j})$. The specific descriptions of these decoders will be shown later on, but for now we are only interested in their definition. The random pmf induced by the random binning schemes can be expressed as:
	 
	\begin{equation}
	 	\begin{aligned} \label{eq:44}
	 		P(x^{n},z^{n}_{1},z^{n}_{2},y^{n}_{[0:2]},&m^{n}_{[0:2]},f_{[0:2]},\hat{h}_{1},\hat{h}_{2})  
	 		\\ = & \phi(x^{n},z^{n}_{1},z^{n}_{2})p(y^{n}_{[0:2]}|x^{n})P(m_{0},f_{0}|y^{n}_{0})P(m_{1},f_{1}|y^{n}_{1})P(m_{2},f_{2}|y^{n}_{2})
	 		\\& \times
	 		P^{HT}(\hat{h}_{1}|m_0,f_0,m_1,f_1,z^{n}_{1})P^{HT}(\hat{h}_{2}|m_0,f_0,m_2,f_2,z^{n}_{2})  \\ = &
	 		P(f_{[0:2]},x^{n},z^{n}_{1},z^{n}_{2})P(y^{n}_{[0:2]}|x^{n},f_{[0:2]})P(m_{0}|y^{n}_{0})P(m_{1}|y^{n}_{1})P(m_{2}|y^{n}_{2}) 
	 		\\& \times
	 		P^{HT}(\hat{h}_{1}|m_0,f_0,m_1,f_1,z^{n}_{1})P^{HT}(\hat{h}_{2}|m_0,f_0,m_2,f_2,z^{n}_{2}),	 
 		\end{aligned} 
	\end{equation}	 

 	where $\phi$ is an indeterminate pmf that would be interpreted as $\phi=\prod_{i=1}^{n}p_{X,Z_1,Z_2}$ in case of the null hypothesis and $\phi=\prod_{i=1}^{n}q_{X,Z_1,Z_2}$ in case of the alternative hypothesis. This setup is illustrated in Figure \ref{ProtocolA}.
 	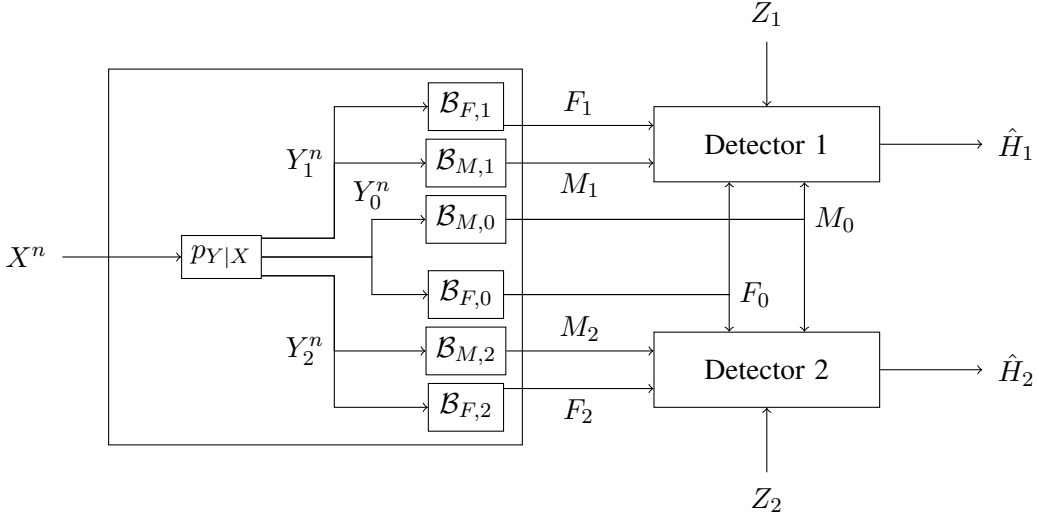
\begin{figure}
 		\centering
 		\begin{tikzpicture}
 			
 			\node (Tr) at (-1,0) [shape=rectangle, draw, minimum width=5.5cm, minimum height=5cm] {};
 			
 			\node (Randomizer) at (-2.25,0) [shape=rectangle, draw, minimum width=1cm, minimum height=0.5cm] {$p_{Y|X}$};
 			
 			\node (B0) at (1,0.5) [shape=rectangle, draw, minimum width=1cm, minimum height=0.5cm] {$\mathcal{B}_{M,0}$};
 			\node (BF0) at (1,-0.5) [shape=rectangle, draw, minimum width=1cm, minimum height=0.5cm] {$\mathcal{B}_{F,0}$};
 			\node (B1) at (1,1.25) [shape=rectangle, draw, minimum width=1cm, minimum height=0.5cm] {$\mathcal{B}_{M,1}$};
 			\node (BF1) at (1,2) [shape=rectangle, draw, minimum width=1cm, minimum height=0.5cm] {$\mathcal{B}_{F,1}$};
 			\node (B2) at (1,-1.25) [shape=rectangle, draw, minimum width=1cm, minimum height=0.5cm] {$\mathcal{B}_{M,2}$};
 			\node (BF2) at (1,-2) [shape=rectangle, draw, minimum width=1cm, minimum height=0.5cm] {$\mathcal{B}_{F,2}$};
 			
 			\node (Rx1) at (5,1.5) [shape=rectangle, draw, minimum width=3cm, minimum height=1cm] {Detector 1};
 			\node (Rx2) at (5,-1.5) [shape=rectangle, draw, minimum width=3cm, minimum height=1cm] {Detector 2};
 			
 			\draw [->] (B1.east) -- node[below] {$M_1$} ($(Rx1.west)+(0,-0.25)$);
 			\draw [->] (B2.east) -- node[above] {$M_2$} ($(Rx2.west)+(0,0.25)$);
 			
 			\draw [->] ($(BF1.east)+(0,-0.25)$) -- node[above] {$F_1$} ($(Rx1.west)+(0,0.25)$);
 			\draw [->] ($(BF2.east)+(0,0.25)$) -- node[below] {$F_2$} ($(Rx2.west)+(0,-0.25)$);
 			
 			\draw [->] (Randomizer.east) -- ($(Randomizer)+(2,0)$) |- node[above] {$Y^{n}_{0}$} (B0.west);
 			\draw [->] (Randomizer.east) -- ($(Randomizer)+(2,0)$) |- (BF0.west);

 			\draw [->] (Randomizer.east) ++(0,0.25) -- ($(Randomizer)+(1.5,0.25)$) |- node[left] {$Y^{n}_{1}$} (B1.west);
 			\draw [->] (Randomizer.east) ++(0,0.25) -- ($(Randomizer)+(1.5,0.25)$) |- (BF1.west);

 			\draw [->] (Randomizer.east) ++(0,-0.25) -- ($(Randomizer)+(1.5,-0.25)$) |- node[left] {$Y^{n}_{2}$} (B2.west);
 			\draw [->] (Randomizer.east) ++(0,-0.25) -- ($(Randomizer)+(1.5,-0.25)$) |- (BF2.west);

 			\draw [->] (B0.east) -| node[right] {$M_0$} ($(Rx1.south)+(0.5,0)$);
 			\draw [->] (B0.east) -| ($(Rx2.north)+(0.5,0)$);
 			
 			\draw [->] (BF0.east) -| node[right] {$F_0$} ($(Rx1.south)+(-0.5,0)$);
 			\draw [->] (BF0.east) -| ($(Rx2.north)+(-0.5,0)$);
 			
 			\node (source0) at ($(Tr)+(-3.5,0)$) [label={[label distance=-2mm]left: {$X^{n}$}}] {};  
 			\draw [->] (source0) -- (Randomizer.west) {};
 			
 			\node (res1) at ($(Rx1)+(3,0)$) [label={[label distance=-2mm]right: {$\hat{H}_{1}$}}] {};  
 			\node (res2) at ($(Rx2)+(3,0)$) [label={[label distance=-2mm]right: {$\hat{H}_{2}$}}] {};  
 			\draw [->] (Rx1.east) -- (res1);
 			\draw [->] (Rx2.east) -- (res2);
 			
 			\node (source1) at ($(Rx1)+(0,1.5)$) [label={[label distance=-2mm]above: {$Z_{1}$}}] {};  
 			\node (source2) at ($(Rx2)+(0,-1.5)$) [label={[label distance=-2mm]below: {$Z_{2}$}}] {};  
 			\draw [->] (source1.south) -- (Rx1.north);
 			\draw [->] (source2.north) -- (Rx2.south);
 			
 		\end{tikzpicture}
 		\caption{Source coding side of the problem (\textit{Protocol A})}
 		 \label{ProtocolA}
 	\end{figure}

	\textit{Protocol B (coding for the main problem assisted with the shared randomness)}: As shown in Figure \ref{ProtocolB}, consider the GWP setup, except for a slight adjustment that both the observer and Detector $j \in \{1,2\}$ have access to a shared randomness $(F_0,F_j)$ where $F_0$ and $F_j$ are uniformly distributed on $[1:2^{n\tilde{R}_0}]$ and $[1:2^{n\tilde{R}_j}]$, respectively. The encoder of the observer acts as follows:
	\begin{enumerate}
		\item The encoder first generates $(Y^{n}_{0},Y^{n}_{1},Y^{n}_{2})$ according to the conditional pmf $P(y^{n}_{[0:2]}|x^{n},f_{[0:2]})$ of $\textit{Protocol A}$.
		\item Subsequently, having obtained $(x^{n},y^{n}_{0},y^{n}_{1},y^{n}_{2})$, the encoder generates index $m_i$ for $i \in [0:2]$ which is the bin index of $y^{n}_i$. To generate the indices, for each $i \in [0:2]$, a random binning scheme maps each sequence $y^{n}_{i}$ to an index according to the conditional pmf $P(m_{i}|y^{n}_{i})$ of $\textit{Protocol A}$. 
		\item Finally, the encoder sends $(M_0,M_1)$ to the first detector and $(M_0,M_2)$ to the second detector. We assume that both the detectors have access to the exact type index of the $(X^{n},Y^{n}_{[0:2]}) \in \mathcal{P}_{n}(\mathcal{X} \times \mathcal{Y}_{[0:2]})$. Because $\abs{\mathcal{P}_{n}(\mathcal{X} \times \mathcal{Y}_{[0:2]})} \leq (n+1)^{\abs{\mathcal{X}} \times \abs{\mathcal{Y}_{[0:2]}}}$, the observer can send the index to detectors through a common zero-rate channel. We denote this index by a random variable $T$ and its realization by $t$.
	\end{enumerate}
\begin{figure}
	\centering
	\begin{tikzpicture}
		
		\node (Tr) at (-1,0) [shape=rectangle, draw, minimum width=5.5cm, minimum height=4cm] {};
		
		\node (Randomizer) at (-2.25,0) [shape=rectangle, draw, minimum width=1cm, minimum height=0.5cm] {$P_{Y^{n}_{[0:2]}|X^{n},F_{[0:2]}}$};
		
		\node (B0) at (1,0) [shape=rectangle, draw, minimum width=0.5cm, minimum height=0.5cm] {$\mathcal{B}_{M,0}$};
		\node (B1) at (1,1) [shape=rectangle, draw, minimum width=0.5cm, minimum height=0.5cm] {$\mathcal{B}_{M,1}$};
		\node (B2) at (1,-1) [shape=rectangle, draw, minimum width=0.5cm, minimum height=0.5cm] {$\mathcal{B}_{M,2}$};
		
		\node (Rx1) at (5,1) [shape=rectangle, draw, minimum width=3cm, minimum height=1cm] {Detector 1};
		\node (Rx2) at (5,-1) [shape=rectangle, draw, minimum width=3cm, minimum height=1cm] {Detector 2};
		
		\draw [->] (B1.east) -- node[above] {$M_1$} (Rx1.west);
		\draw [->] (B2.east) -- node[below] {$M_2$} (Rx2.west);
		
		\draw [->] (Randomizer.east) -- node[above] {$Y^{n}_{0}$} (B0.west);
		\draw [->] (Randomizer.east) ++(0,0.25) -- ($(Randomizer)+(1.5,0.25)$) |- node[above] {$Y^{n}_{1}$} (B1.west);
		\draw [->] (Randomizer.east) ++(0,-0.25) -- ($(Randomizer)+(1.5,-0.25)$) |- node[below] {$Y^{n}_{2}$} (B2.west);
		
		\draw [<->,dash dot] (Randomizer.north) |- ($(Randomizer)+(0,2.5)$) -- ($(Randomizer)+(3.25,2.5)$) node[above] {$F_0,F_1$} -| ($(Rx1.north)+(-1,0)$);
		\draw [<->,dash dot] (Randomizer.south) |- ($(Randomizer)+(0,-2.5)$) -- ($(Randomizer)+(3.25,-2.5)$) node[below] {$F_0,F_2$} -| ($(Rx2.south)+(-1,0)$);
		
		\draw [->] (B0.east) -| node[right] {$M_0$} (Rx1.south);
		\draw [->] (B0.east) -| (Rx2.north);
		
		\node (source0) at ($(Tr)+(-3.5,0)$) [label={[label distance=-2mm]left: {$X^{n}$}}] {};  
		\draw [->] (source0) -- (Randomizer.west) {};
		
		\node (res1) at ($(Rx1)+(3,0)$) [label={[label distance=-2mm]right: {$\hat{H}_{1}$}}] {};  
		\node (res2) at ($(Rx2)+(3,0)$) [label={[label distance=-2mm]right: {$\hat{H}_{2}$}}] {};  
		\draw [->] (Rx1.east) -- (res1);
		\draw [->] (Rx2.east) -- (res2);
		
		\node (source1) at ($(Rx1)+(0,1.5)$) [label={[label distance=-2mm]above: {$Z_{1}$}}] {};  
		\node (source2) at ($(Rx2)+(0,-1.5)$) [label={[label distance=-2mm]below: {$Z_{2}$}}] {};  
		\draw [->] (source1.south) -- (Rx1.north);
		\draw [->] (source2.north) -- (Rx2.south);
		
	\end{tikzpicture}
	\caption{Main problem assisted with a shared randomness (\textit{Protocol B})}
	\label{ProtocolB}
\end{figure}
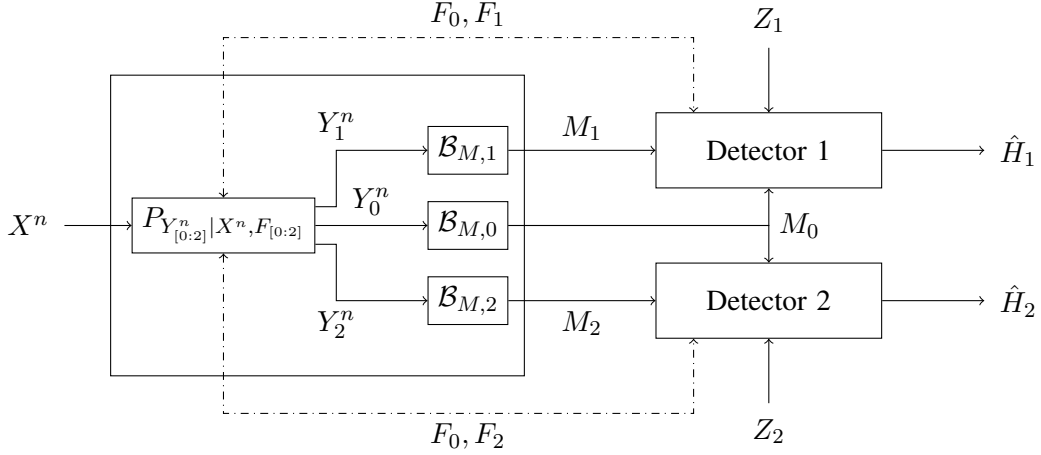

	Detector $j \in \{1,2\}$ performs the hypothesis testing employing decoder $P^{HT}(\hat{h}_{j}|m_0,f_0,m_j,f_j,z^{n}_{j})$ of $\textit{Protocol A}$. The random pmf induced by this protocol, denoted as $\hat{P}$, can be expressed as
	
	\begin{equation} \label{eq:43}
		\begin{aligned}
			\hat{P}(x^{n},z^{n}_{1},z^{n}_{2},y^{n}_{[0:2]},&m^{n}_{[0:2]},f_{[0:2]},\hat{h}_{1},\hat{h}_{2})  \\ =&
			p^{U}(f_{[0:2]})\phi(x^{n},z^{n}_{1},z^{n}_{2})P(y^{n}_{[0:2]}|x^{n},f_{[0:2]})P(m_{0}|y^{n}_{0})P(m_{1}|y^{n}_{1}) 
			\\& \times
			P(m_{2}|y^{n}_{2})P^{HT}(\hat{h}_{1}|m_0,f_0,m_1,f_1,z^{n}_{1})P^{HT}(\hat{h}_{2}|m_0,f_0,m_2,f_2,z^{n}_{2}),   
		\end{aligned}
	\end{equation}

	Note that $P(y^{n}_{[0:2]}|x^{n},f_{[0:2]})$ is independent of the $\phi$ as long as the consistency condition on the marginal distributions of the two hypotheses holds, because it can be displayed as
	\[
		P(y^{n}_{[0:2]}|x^{n},f_{[0:2]}) = 
		\frac{P(y^{n}_{[0:2]},x^{n},f_{[0:2]})}{\sum_{y^{n}_{[0:2]} \in \mathcal{Y}^{n}_{[0:2]}}P(y^{n}_{[0:2]},x^{n},f_{[0:2]})} =
		\frac{\phi(x^{n})p(y^{n}_{[0:2]}|x^{n})P(f_{[0:2]}|y^{n}_{[0:2]})}{\sum_{y^{n}_{[0:2]} \in \mathcal{Y}^{n}_{[0:2]}}\phi(x^{n})p(y^{n}_{[0:2]}|x^{n})P(f_{[0:2]}|y^{n}_{[0:2]})},
	\]
	which is indifferent towards the particular occurrence of $\phi$ since $\phi(x^n)=\prod_{i=1}^{n}p_{X}(x^n)$ is valid regardless of the true hypothesis.
	
	\proofpart{2a}{Sufficient conditions that make the hypothesis testing in the dual setup successful.}
	We deem a hypothesis testing scheme successful when the obtained type \rom{1} error by the scheme is vanishing and the type \rom{2} error fades exponentially as $n \rightarrow \infty$. For this evaluation to be made, first we need to describe our proposed hypothesis testing scheme at the detectors.  
	For Detector $j \in \{1,2\}$ consider the following events: 
	\begin{equation*}
	\begin{gathered}
		\mathcal{E}_{0}=\left\{T \in \mathcal{T}^{n}_{[\, p_{X,Y_{[0:2]}}]_{\delta^{\prime}_{n}}}\right\}, \\ 
		\mathcal{E}_{j}\!=\!\left\{ \exists \ (\tilde{y}^{n}_{0},\tilde{y}^{n}_{j}) : \mathcal{B}_{M,0}(\tilde{y}^{n}_{0}) \!=\! M_0 \wedge \mathcal{B}_{F,0}(\tilde{y}^{n}_{0})\!=\!F_0 \wedge \mathcal{B}_{M,j}(\tilde{y}^{n}_{j})\!=\!M_j \wedge \mathcal{B}_{F,j}(\tilde{y}^{n}_{j})\!=\!F_j \wedge   (\tilde{y}^{n}_{0},\tilde{y}^{n}_{j},Z^{n}_{j}) \in \mathcal{T}^{n}_{[\, p_{Y_0,Y_j,Z_j}]_{\delta^{\prime}_{n}}} \right\},		
	\end{gathered}
	\end{equation*}
	 where $p_{Y_0,Y_j,Z_j}(y^{n}_{0},y^{n}_{j},z^{n}_{j})=\sum_{x^{n}}p(x^n,z^{n}_{j})p(y^{n}_{0},y^{n}_{j}|x^{n})$. Note that if we define 
	 
	 \begin{equation*}
	 	\begin{gathered}
	 		\mathcal{E}_{j,S}= \left\{(Y^{n}_{0},Y^{n}_{j},Z^{n}_{j}) \in \mathcal{T}^{n}_{[\, p_{Y_0,Y_j,Z_j}]_{\delta^{\prime}_{n}}}\right\}, \\
	 		\mathcal{E}_{j,NS}=
	 		\left\{\begin{array}{cc}
	 			 \exists \ (\tilde{y}^{n}_{0},\tilde{y}^{n}_{j}) : \mathcal{B}_{M,0}(\tilde{y}^{n}_{0}) = M_0 \wedge \mathcal{B}_{F,0}(\tilde{y}^{n}_{0})=F_0 \wedge \mathcal{B}_{M,j}(\tilde{y}^{n}_{J})=M_j \wedge \mathcal{B}_{F,j}(\tilde{y}^{n}_{j})=F_j \\ 
	 			\text{for some } \tilde{y}^{n}_{0} \neq Y^{n}_{0} \text{ or } \tilde{y}^{n}_{j} \neq Y^{n}_{j} \text{ such that }  (\tilde{y}^{n}_{0},\tilde{y}^{n}_{j},Z^{n}_{j}) \in \mathcal{T}^{n}_{[\, p_{Y_0,Y_j,Z_j}]_{\delta^{\prime}_{n}}} 
	 		\end{array}\right\},
	 	\end{gathered}
	 \end{equation*}
 	evidently we have $\mathcal{E}_{j}=\{\mathcal{E}_{j,S} \cup \mathcal{E}_{j,NS}\}$. The decision function at Detector $j \in \{1,2\}$ can be expressed as follows:
	 \begin{equation}
	 	\hat{H}_{j} =  g^{n}_{j}(Z^{n}_{1},M_0,M_j,F_0,F_j) = 1 - \mathbbm{1}\left(
	 	\mathcal{E}_{0}
	 	\cap
	 	\mathcal{E}_{j}
	 	\right).
	 \end{equation}
	\textit{Type \rom{1} error analysis}: The following lemma describes a vanishing upper bound for the type \rom{1} error of the dual problem:
	\begin{lemma}\thlabel{type1dual}
		The type \rom{1} error of the HT at Detector $j \in \{1,2\}$ of \textit{Protocol A} is bounded as:
		\begin{equation}
			P\condx*{\hat{H}_{j}=1}{H=0} \leq \epsilon_{n} + \delta_{n} \stackrel{n \rightarrow \infty}{\rightarrow} 0 \label{eq:57}.
		\end{equation}
	\end{lemma}
	\begin{proof}
	Consider the case where the true hypothesis corresponds to the null hypothesis, $H=0$, implying that $\phi_{X,Z_1,Z_2}=p_{X,Z_1,Z_2}$. The type \rom{1} error of the hypothesis testing at Detector $j \in \{1,2\}$ of the dual protocol in this case can be written as follows: 
	\begin{align}
		P\condx*{\hat{H}_{j}=1}{H=0} 
		&= P\condx*{\mathcal{E}^{c}_{0} \cup \mathcal{E}^{c}_{j}}{ \phi_{X,Z_1,Z_2}=p_{X,Z_1,Z_2}} \\
		&\leq 
		P\condx*{\mathcal{E}^{c}_{0}}{ \phi_{X,Z_1,Z_2}=p_{X,Z_1,Z_2}} +
		P\condx*{\mathcal{E}^{c}_{j}}{ \phi_{X,Z_1,Z_2}=p_{X,Z_1,Z_2}}, \label{eq:47}
	\end{align}
	where \eqref{eq:47} follows from the union bound, also known as Boole's inequality. 
	Recall that the tuple $(X^{n},Y^{n}_{[0:2]},Z^{n}_{[1:2]})$ in the dual problem is i.i.d according to $p_{X,Z_1,Z_2}p_{Y_{[0;2]}|X}$, meaning the terms in \eqref{eq:47} could be bounded as 
	\begin{equation}
		P\condx*{\mathcal{E}^{c}_{0}}{ \phi_{X,Z_1,Z_2}=p_{X,Z_1,Z_2}} \leq \epsilon_n \rightarrow 0, \label{eq:48}
	\end{equation}
	and
	\begin{align}
	P\condx*{\mathcal{E}^{c}_{j}}{ \phi_{X,Z_1,Z_2}=p_{X,Z_1,Z_2}} &\leq P\condx*{
	\mathcal{E}^{c}_{j,S}	\cap \mathcal{E}^{c}_{j,NS}
	}{ \phi_{X,Z_1,Z_2}=p_{X,Z_1,Z_2}} \\
	&\leq P\condx*{
		\mathcal{E}^{c}_{j,S}
	}{ \phi_{X,Z_1,Z_2}=p_{X,Z_1,Z_2}} \\
	&\leq \delta_n \rightarrow 0, \label{eq:49}
	\end{align}
	where \eqref{eq:48} and \eqref{eq:49} follow from the AEP. Subsequently we obtain:
	\begin{equation}
		P\condx*{\hat{H}_{j}=1}{H=0} \leq \epsilon_{n} + \delta_{n} \rightarrow 0.
	\end{equation}
	\end{proof}
	\textit{Type \rom{2} error analysis}: 
	When the true hypothesis is $H=1$, meaning $\phi_{X,Z_1,Z_2}=q_{X,Z_1,Z_2}$, the type \rom{2} error at Detector $j \in \{1,2\}$ of the dual protocol is evaluated in the following lemma.
	\begin{lemma}\thlabel{type2dual}
		In \textit{Protocol A} for dual problem, if $(R_{[0:2]},\tilde{R}_{[0:2]})$ satisfy the following conditions:
		\begin{equation} 
			\begin{gathered}
				R_{0} + \tilde{R}_{0} > H(Y_{0}|Y_{j}Z_{j}), \\
				R_{j} + \tilde{R}_{j} > H(Y_{j}|Y_{0}Z_{j}), \\
				R_{0} + \tilde{R}_{0} + R_{j} + \tilde{R}_{j} > H(Y_{0}Y_{j}|Z_{j}), \\
			\end{gathered}
		\end{equation}
		then the type \rom{2} error of the HT at Detector $j \in \{1,2\}$ is bounded as:
		\begin{align}
			-\frac{1}{n} \limsup_{n \rightarrow \infty} 
			P\condx*{\hat{H}_{j}=0}{H=1} \geq
			E_{0,j}(p_{Y_{[0:2]}|X}) +
			E_{1,j}(p_{Y_{[0:2]}|X}). \label{eq:61}
		\end{align}
	\end{lemma}
\begin{proof}
	We expand the type \rom{2} error at Detector $j \in \{1,2\}$ as:
	\begin{align}
		P\condx*{\hat{H}_{j}=0}{H=1}\!
		&=\! P\condx*{\mathcal{E}_{0} \cap \mathcal{E}_{j}}{ \phi_{X,Z_1,Z_2}=q_{X,Z_1,Z_2}} \\
		&= \! P\condx*{\mathcal{E}_{0} \cap \{\mathcal{E}_{j,S} \cup \mathcal{E}_{j,NS}\}}{ \phi_{X,Z_1,Z_2}=q_{X,Z_1,Z_2}} \\
		&\leq		
		\! P\condx*{\mathcal{E}_{0} \cap \mathcal{E}_{j,S}}{ \phi_{X,Z_1,Z_2}=q_{X,Z_1,Z_2}}\! +\!
		P\condx*{\mathcal{E}_{0} \cap \mathcal{E}_{j,NS}}{ \phi_{X,Z_1,Z_2}=q_{X,Z_1,Z_2}}, \label{eq:50}
	\end{align}
		where \eqref{eq:50} follows from the union bound. From now on we are using the $P\condx*{\cdot}{ \phi_{X,Z_1,Z_2}=q_{X,Z_1,Z_2}} \coloneqq P_q(\cdot)$ for the sake of convenience. Note that the tuple $(X^{n},Y^{n}_{[0:2]},Z^{n}_{[1:2]})$ in the dual problem is i.i.d according to $q_{X,Z_1,Z_2}p_{Y_{[0;2]}|X}$, permitting the use of Sanov's theorem \cite[Problem 2.12]{csiszar11} to bound the first term in \eqref{eq:50} as follows,
		
		\begin{align}
			-\frac{1}{n}\log P_q\para*{\mathcal{E}_{0} \cap \mathcal{E}_{j,S}}
			\geq
			\min_{\pi_{\bar{X},\bar{Y}_{[0:2]},\bar{Z}_j} \in \mathcal{K}^{n}_{0,j}} \infdiv*{\pi_{\bar{X},\bar{Y}_{[0:2]},\bar{Z}_j}}{q_{X,Z_j}p_{Y_{[0:2]}|X}} \label{eq:60}
			- \nu_{n,j},
		\end{align}
	for 
	\[
		\mathcal{K}^{n}_{0,j}= \left\{
		\pi_{\bar{X},\bar{Y}_{[0:2]},\bar{Z}_j} \in \mathcal{P}_{n}\left(\mathcal{X}\times\mathcal{Y}_{[0:2]}\times\mathcal{Z}_j\right) \ : \
		\pi_{\bar{X},\bar{Y}_{[0:2]}} \stackrel{\delta^{\prime}_{n}}{\approx} p_{X,Y_{[0:2]}} \wedge
		\pi_{\bar{Y}_{0},\bar{Y}_{j},\bar{Z}_j} \stackrel{\delta^{\prime}_{n}}{\approx} p_{Y_{0},Y_{j},Z_j}
		\right\},
	\]
where $p_{Y_0,Y_j,Z_j}(y^{n}_{0},y^{n}_{j},z^{n}_{j})=\sum_{x^{n}}p(x^n,z^{n}_{j})p(y^{n}_{0},y^{n}_{j}|x^{n})$ and $\nu_{n,j} \coloneqq \frac{\log{n+1}}{n}\abs{\mathcal{X}}\abs{\mathcal{Y}_{[0:2]}}\abs{\mathcal{Z}_j}$. We refer to this obtained exponent by
 \[
 	E^{n}_{0,j}(p_{Y_{[0:2]}|X}) \coloneqq 
 	\min_{\pi_{\bar{X},\bar{Y}_{[0:2]},\bar{Z}_j} \in \mathcal{K}^{n}_{0,j}} \infdiv*{\pi_{\bar{X},\bar{Y}_{[0:2]},\bar{Z}_j}}{q_{X,Z_j}p_{Y_{[0:2]}|X}}
 	- \nu_{n,j}.
 \]

For the second term in \eqref{eq:50} we can further decouple the events by forming new combinations as:‌
	\begin{equation*}
		\begin{gathered}
		\mathcal{E}_{j,NS,0}=
		\left\{\begin{array}{cc}
			\exists \ (\tilde{y}^{n}_{0},\tilde{y}^{n}_{j}) : \mathcal{B}_{M,0}(\tilde{y}^{n}_{0}) = M_0 \wedge \mathcal{B}_{F,0}(\tilde{y}^{n}_{0})=F_0 \wedge \mathcal{B}_{M,j}(\tilde{y}^{n}_{J})=M_j \wedge \mathcal{B}_{F,j}(\tilde{y}^{n}_{j})=F_j \\ 
			\text{for some } \tilde{y}^{n}_{0} \neq Y^{n}_{0} \text{ and } \tilde{y}^{n}_{j} \neq Y^{n}_{j} \text{ such that }  (\tilde{y}^{n}_{0},\tilde{y}^{n}_{j},Z^{n}_{j}) \in \mathcal{T}^{n}_{[\,p_{Y_0,Y_j,Z_j}]_{\delta^{\prime}_{n}}} 
		\end{array}\right\}, \\
	\mathcal{E}_{j,NS,1}=
	\left\{\begin{array}{cc}
		\exists \ \tilde{y}^{n}_{0} : \mathcal{B}_{M,0}(\tilde{y}^{n}_{0}) = M_0 \wedge \mathcal{B}_{F,0}(\tilde{y}^{n}_{0})=F_0  \\ 
		\text{for some } \tilde{y}^{n}_{0} \neq Y^{n}_{0} \text{ such that }  (\tilde{y}^{n}_{0},Y^{n}_{j},Z^{n}_{j}) \in \mathcal{T}^{n}_{[\,p_{Y_0,Y_j,Z_j}]_{\delta^{\prime}_{n}}} 
	\end{array}\right\}, \\
\mathcal{E}_{j,NS,2}=
\left\{\begin{array}{cc}
	\exists \ \tilde{y}^{n}_{j} : \mathcal{B}_{M,j}(\tilde{y}^{n}_{j}) = M_j \wedge \mathcal{B}_{F,j}(\tilde{y}^{n}_{j})=F_j  \\ 
	\text{for some } \tilde{y}^{n}_{j} \neq Y^{n}_{j} \text{ such that }  (Y^{n}_{0},\tilde{y}^{n}_{j},Z^{n}_{j}) \in \mathcal{T}^{n}_{[\,p_{Y_0,Y_j,Z_j}]_{\delta^{\prime}_{n}}} 
\end{array}\right\}.
	\end{gathered}
\end{equation*}
Note that $\mathcal{E}_{j,NS}=\left\{ \mathcal{E}_{j,NS,0} \cup \mathcal{E}_{j,NS,1} \cup \mathcal{E}_{j,NS,2} \right\}$. By using the union bound, one can write the second term in \eqref{eq:50} as:

	\begin{align}
		P_q\para*{\mathcal{E}_{0} \cap \mathcal{E}_{j,NS}}
		& \leq
		\sum_{i \in [0:2]}
		P_q\para*{\mathcal{E}_{0} \cap \mathcal{E}_{j,NS,i}}
		\\& =
		\sum_{i \in [0:2]}
		P_q\para*{\mathcal{E}_{0} \cap \Psi_{i,j}}
		P_q\condx*{\mathcal{E}_{j,NS,i}}{\mathcal{E}_{0} \wedge \Psi_{i,j}}, \label{eq:51}
	\end{align}
	where $\Psi_{0,j}\coloneqq\left\{Z^{n}_j \in \mathcal{T}^{n}_{[\,p_{Z_j}]_{\delta^{'}_{n}}}\right\}$, $\Psi_{1,j}\coloneqq\left\{(Y^{n}_j,Z^{n}_j) \in \mathcal{T}^{n}_{[\,p_{Y_j,Z_j}]_{\delta^{'}_{n}}} \right\}$ and $\Psi_{2,j}\coloneqq\left\{(Y^{n}_0,Z^{n}_j) \in \mathcal{T}^{n}_{[\,p_{Y_0,Z_j}]_{\delta^{'}_{n}}} \right\}$.
	The first term inside the sum in \eqref{eq:51} can be bounded again by using Sanov's theorem, \cite[Problem 2.12]{csiszar11}, yielding the following results for $i \in [0:2]$,
	\begin{align}
		-\frac{1}{n}\log P_q\para*{\mathcal{E}_{0} \cap \Psi_{i,j}}
		\geq
		\min_{\pi_{\bar{X},\bar{Y}_{[0:2]},\bar{Z}_j} \in \mathcal{K}^{n}_{\Psi_{i,j}}} \infdiv*{\pi_{\bar{X},\bar{Y}_{[0:2]},\bar{Z}_j}}{q_{X,Z_j}p_{Y_{[0:2]}|X}} \label{eq:80}
		- \nu_{n,j},
	\end{align}
	where, 
	\begin{equation*}
		\begin{gathered}
	\mathcal{K}^{n}_{\Psi_{0,j}}= \left\{
	\pi_{\bar{X},\bar{Y}_{[0:2]},\bar{Z}_j} \in \mathcal{P}_{n}\left(\mathcal{X}\times\mathcal{Y}_{[0:2]}\times\mathcal{Z}_j\right) \ : \
	\pi_{\bar{X},\bar{Y}_{[0:2]}} \stackrel{\delta^{\prime}_{n}}{\approx} p_{X,Y_{[0:2]}} \wedge
	\pi_{\bar{Z}_j} \stackrel{\delta^{\prime}_{n}}{\approx} p_{Z_j}
	\right\}, \\
	\mathcal{K}^{n}_{\Psi_{1,j}}= \left\{
	\pi_{\bar{X},\bar{Y}_{[0:2]},\bar{Z}_j} \in \mathcal{P}_{n}\left(\mathcal{X}\times\mathcal{Y}_{[0:2]}\times\mathcal{Z}_j\right) \ : \
	\pi_{\bar{X},\bar{Y}_{[0:2]}} \stackrel{\delta^{\prime}_{n}}{\approx} p_{X,Y_{[0:2]}} \wedge
	\pi_{\bar{Y}_j,\bar{Z}_j} \stackrel{\delta^{\prime}_{n}}{\approx} p_{Y_j,Z_j}
	\right\}, \\
	\mathcal{K}^{n}_{\Psi_{2,j}}= \left\{
	\pi_{\bar{X},\bar{Y}_{[0:2]},\bar{Z}_j} \in \mathcal{P}_{n}\left(\mathcal{X}\times\mathcal{Y}_{[0:2]}\times\mathcal{Z}_j\right) \ : \
	\pi_{\bar{X},\bar{Y}_{[0:2]}} \stackrel{\delta^{\prime}_{n}}{\approx} p_{X,Y_{[0:2]}} \wedge
	\pi_{\bar{Y}_0,\bar{Z}_j} \stackrel{\delta^{\prime}_{n}}{\approx} p_{Y_0,Z_j}
	\right\},
\end{gathered}
\end{equation*}
	and $\nu_{n,j} = \frac{\log{n+1}}{n}\abs{\mathcal{X}}\abs{\mathcal{Y}_{[0:2]}}\abs{\mathcal{Z}_j}$, as we defined earlier.
	To bound the $P_q\condx*{\mathcal{E}_{j,NS,0}}{\mathcal{E}_{0} \wedge \Psi_{0,j}}$ in \eqref{eq:51}, one can easily see that while the particular instances of $(y^{n}_{0},y^{n}_{j})$ who are jointly typical with $z^{n}_j$ depend strictly on the specific choice of $z^{n}_j$, their number, i.e., $\abs{\mathcal{T}^{n}_{p_{Y_0,Y_j|Z_j}}(z^{n}_j)}$, depends on $z^{n}_j$ only through their type. Consequently, the probability of $\mathcal{E}_{j,NS,0}$ can be bounded by using the law of total probability as:
	
	\begin{align}
		P_q&\condx*{\mathcal{E}_{j,NS,0}}{\mathcal{E}_{0} \wedge \Psi_{0,j}} =   
		\frac{P_q\condx*{\mathcal{E}_{j,NS,0} \cap \Psi_{0,j}}{ \mathcal{E}_{0}}}{P_q\condx*{\Psi_{0,j}}{\mathcal{E}_{0}}}
		\label{eq:68}
		\\ 
		&\begin{aligned}
		&=
		\sum_{\mathclap{z^{n}_{j} \in \Psi_{0,j}}}\frac{q\condx*{z^{n}_{j}}{\mathcal{E}_{0}}}{P_q\condx*{\Psi_{0,j}}{\mathcal{E}_{0}}} \\& \times
		\osum{(y^{n}_{0},y^{n}_{j}) \in \mathcal{T}^{n}_{[\,p_{Y_0,Y_j|Z_j}]_{\delta^{'}_{n}}}\!(z^{n}_j)}{
		P_q\left(		
		 \mathcal{B}_{M_0}(y^{n}_{0})\!=\!M_0 \wedge \mathcal{B}_{F_0}(y^{n}_{0})\!=\!F_0 \wedge \mathcal{B}_{M_j}(y^{n}_{j})\!=\!M_j \wedge \mathcal{B}_{F_j}(y^{n}_{j})\!=\!F_j
		\right)}  \label{eq:69}
		\end{aligned}
		\\&=
		\sum_{z^{n}_{j} \in \Psi_{0,j}}\frac{q\condx*{z^{n}_{j}}{\mathcal{E}_{0}}}{P_q\condx*{\Psi_{0,j}}{\mathcal{E}_{0}}}
		\sum_{(y^{n}_{0},y^{n}_{j}) \in \mathcal{T}^{n}_{[\,p_{Y_0,Y_j|Z_j}]_{\delta^{'}_{n}}}\!(z^{n}_j)}		
		2^{-nR_0} \times 2^{-n\tilde{R}_0} \times 2^{-nR_j} \times 2^{-n\tilde{R}_j}
		\label{eq:70}
		\\&=
		\sum_{z^{n}_{j} \in \Psi_{0,j}}\frac{q\condx*{z^{n}_{j}}{\mathcal{E}_{0}}}{P_q\condx*{\Psi_{0,j}}{\mathcal{E}_{0}}}
		\abs{\mathcal{T}^{n}_{[\,p_{Y_0,Y_j|Z_j}]_{\delta^{'}_{n}}}\!(z^{n}_j)}		
		2^{-n\left(R_0 + \tilde{R}_0 + R_j + \tilde{R}_j\right)}
		\label{eq:71}
		\\ &\leq
		\sum_{z^{n}_{j} \in \Psi_{0,j}}\frac{q\condx*{z^{n}_{j}}{\mathcal{E}_{0}}}{P_q\condx*{\Psi_{0,j}}{\mathcal{E}_{0}}} \times
		2^{n\left(
			\cond{Y_0,Y_j}{Z_j}
			+ \eta_n
			\right)}		\times
		2^{-n\left(R_0 + \tilde{R}_0 + R_j + \tilde{R}_j\right)} 
		\label{eq:72}
		\\ &\leq
		2^{-n\left(
			R_0 + \tilde{R}_0 + R_j + \tilde{R}_j -
			\cond*{Y_0,Y_j}{Z_j}
			- \eta_{0,n}
		 \right)}, \label{eq:73}
	\end{align}
	where $\eta_{0,n} \rightarrow 0$ as $n \rightarrow \infty$; \eqref{eq:69} follows from the definition of the events and the law of total probability, \eqref{eq:70} follows because the random mappings are done uniformly and independently, \eqref{eq:71} is reached since the terms inside the summation in \eqref{eq:70} do not depend on the specific values of $(y^{n}_{0},y^{n}_{j})$, and therefore, the summation can be replaced by the size of its subscription, i.e., $\abs{\mathcal{T}^{n}_{[\,p_{Y_0,Y_j|Z_j}]_{\delta^{'}_{n}}}\!(z^{n}_j)}$. This term then is bounded using \cite[Lemma 2.13]{csiszar11} to prompt \eqref{eq:72}. Eventually, inequality in \eqref{eq:73} is attained since by the definition of $\Psi_{0,j}$, it is evident that $\sum_{z^{n}_{j} \in \Psi_{0,j}}\frac{q\condx*{z^{n}_{j}}{\mathcal{E}_{0}}}{P_q\condx*{\Psi_{0,j}}{\mathcal{E}_{0}}} = 1$ for $j \in \{1,2\}$.
	
	By using the same method one can bound $P_q\condx*{\mathcal{E}_{j,NS,i}}{\mathcal{E}_{0} \wedge \Psi_{i,j}}$ for $i \in \{1,2\}$ to obtain the following results.
	
	\begin{align}
		P_q\condx*{\mathcal{E}_{j,NS,1}}{\mathcal{E}_{0} \wedge \Psi_{1,j}} \leq 
		2^{-n\left(
			R_0 + \tilde{R}_0  -
			\cond*{Y_0}{Z_j,Y_j}
			- \eta_{1,n}
			\right)}, \label{eq:78}
		\\
		P_q\condx*{\mathcal{E}_{j,NS,2}}{\mathcal{E}_{0} \wedge \Psi_{2,j}} \leq 
		2^{-n\left(
			R_j + \tilde{R}_j  -
			\cond*{Y_j}{Z_j,Y_0}
			- \eta_{2,n}
			\right)}, \label{eq:79}
	\end{align}
	where $\eta_{j,n} \rightarrow 0$ as $n \rightarrow
	 \infty$ for $j \in \{1,2\}$. The term in \eqref{eq:51} can be bounded as follows.
	
	\begin{align}
		P_q&\para*{\mathcal{E}_{0} \cap \mathcal{E}_{j,NS}}
		 \leq
		\sum_{i \in [0:2]}
		2^{-n \left( \min_{\pi \in \mathcal{K}^{n}_{\Psi_{i,j}}} \infdiv*{\pi_{\bar{X},\bar{Y}_{[0:2]},\bar{Z}_j}}{q_{X,Z_j}p_{Y_{[0:2]}|X}} 
			- \nu_{n,j}
			\right)}
		P_q\condx*{\mathcal{E}_{j,NS,i}}{\mathcal{E}_{0} \wedge \Psi_{i,j}} \label{eq:74}
		\\& \leq
		\max_{i \in [0:2]}
		\left\{
		2^{-n \left( \min_{\pi \in \mathcal{K}^{n}_{\Psi_{i,j}}} \infdiv*{\pi_{\bar{X},\bar{Y}_{[0:2]},\bar{Z}_j}}{q_{X,Z_j}p_{Y_{[0:2]}|X}} 
			- \nu_{n,j}
			\right)}		
		\right\}
		\sum_{i \in [0:2]}
		P_q\condx*{\mathcal{E}_{j,NS,i}}{\mathcal{E}_{0} \wedge \Psi_{i,j}} \label{eq:76}
		\\ &=
		2^{-n \left( \min_{\pi \in \mathcal{K}^{n}_{\Psi_{0,j}}} \infdiv*{\pi_{\bar{X},\bar{Y}_{[0:2]},\bar{Z}_j}}{q_{X,Z_j}p_{Y_{[0:2]}|X}} 
		- \nu_{n,j}
		\right)}
		\sum_{i \in [0:2]}
		P_q\condx*{\mathcal{E}_{j,NS,i}}{\mathcal{E}_{0} \wedge \Psi_{i,j}} \label{eq:75}
		\\& 
		\begin{aligned}
		\leq
		2^{-n \left( \min_{\pi \in \mathcal{K}^{n}_{\Psi_{0,j}}} \infdiv*{\pi_{\bar{X},\bar{Y}_{[0:2]},\bar{Z}_j}}{q_{X,Z_j}p_{Y_{[0:2]}|X}} 
			- \nu_{n,j}
			\right)}&
		\\ \times
		\max_{\emptyset \neq \mathcal{S} \in \{0,j\}}
		&\left\{
		3 \times 
		2^{-n\para*{\sum_{i \in \mathcal{S}}R_i + \tilde{R}_i - \cond{Y_{\mathcal{S}}}{Z_j,Y_{\mathcal{S}^{c}}}}
		- \eta_n} 
	\right\}  \label{eq:77}
	\end{aligned}
		\\& \leq
		2^{-n\left\{ \min_{\pi} \in \mathcal{K}_{\Psi_0} \infdiv*{\pi_{\bar{X},\bar{Y}_{[0:2]},\bar{Z}_j}}{q_{X,Z_j}p_{Y_{[0:2]}|X}}
		+
		\min_{\emptyset \neq \mathcal{S} \in \{0,j\}}
		\para*{\sum_{i \in \mathcal{S}}R_i + \tilde{R}_i - \cond{Y_{\mathcal{S}}}{Z_j,Y_{\mathcal{S}^{c}}}}
	- \kappa_n\right\}}, \label{eq:59}
	\end{align}
	where $\eta_n \coloneqq \max_{i \in [0:2]}(\eta_{i,n})$ and $\kappa_n \coloneqq \frac{\log(3)}{n} +‌\nu_{n,j} +‌ \eta_n$. The inequality in \eqref{eq:74} follows from \eqref{eq:80}. Note that $\mathcal{K}^{n}_{\Psi_{1,j}} \subseteq \mathcal{K}^{n}_{\Psi_{0,j}}$ and $\mathcal{K}^{n}_{\Psi_{2,j}} \subseteq \mathcal{K}^{n}_{\Psi_{0,j}}$, meaning that $\mathcal{K}^{n}_{\Psi_{0,j}}$ results in a larger upper bound than $\mathcal{K}^{n}_{\Psi_{i,j}}$ for $i \in \{1,2\}$, hence \eqref{eq:75} follows. \eqref{eq:77} is resulted from maximizing among the upper bounds achieved in \eqref{eq:73}, \eqref{eq:78}, and \eqref{eq:79}.
	For simplicity, we will use the following convention from now on:
	\[
		E^{n}_{1,j}(p_{Y_{[0:2]}|X})
		\coloneqq
		\min_{\pi_{\bar{X},\bar{Y}_{[0:2]},\bar{Z}_j} \in \mathcal{K}^{n}_{\Psi_{0,j}}} \infdiv*{\pi_{\bar{X},\bar{Y}_{[0:2]},\bar{Z}_j}}{q_{X,Z_j}p_{Y_{[0:2]}|X}}
		+
		\min_{\emptyset \neq \mathcal{S} \in \{0,j\}}
		\para*{\sum_{i \in \mathcal{S}}R_i + \tilde{R}_i - \cond{Y_{\mathcal{S}}}{Z_j,Y_{\mathcal{S}^{c}}}}
		- \kappa_n.
	\]

	Now by combining the results from \eqref{eq:60} and \eqref{eq:59}, we come by the following bound for the dual problem:
	\begin{align}
		P\condx*{\hat{H}_{j}=0}{H=1} \leq
		2^{-nE^{n}_{0,j}(p_{Y_{[0:2]}|X})} +
		2^{-nE^{n}_{1,j}(p_{Y_{[0:2]}|X})}.
	\end{align}
	
	For this error probability to converge to zero exponentially for $j \in \{1,2\}$, we must have:
	\begin{equation} \label{eq:conditionsSecond}
		\begin{gathered}
			R_{0} + \tilde{R}_{0} > H(Y_{0}|Y_{j}Z_{j}), \\
			R_{j} + \tilde{R}_{j} > H(Y_{j}|Y_{0}Z_{j}), \\
			R_{0} + \tilde{R}_{0} + R_{j} + \tilde{R}_{j} > H(Y_{0}Y_{j}|Z_{j}). \\
		\end{gathered}
	\end{equation}
\end{proof}	
	\proofpart{2b}{Sufficient conditions that make the induced pmfs approximately the same}
	Now that we have the necessary bounds regarding the events in the dual problem, we are interested in finding the conditions that make the pmf $P$ close to $\hat{P}$ in terms of total variation distance. By achieving such conditions we can apply those upper bounds to the main problem assisted with the shared randomness. The following lemma provides an upper bound on the total variation distance between random pmfs induced in \textit{Protocol A} and \textit{Protocol A}.
	\begin{lemma}
		For $j \in \{1,2\}$, following bounds could be applied:
	\begin{gather}
		\mathbbm{E}\norm{
			\hat{P}(\cdot|\Psi_{0,j})
			-
			P(\cdot|\Psi_{0,j})
		}_{TV} \label{eq:1000}
		\leq
		2^{-n\aleph(R_{\mathcal{T}},\phi_{X^{n}|\Psi_0},p_{Y_{[0:2]}|X})}, \\
		\mathbbm{E}\norm{
		\hat{P}(\cdot)
		-
		P(\cdot)}_{TV} 
		\leq
	2^{-n\zeta(R_{\mathcal{T}},\phi_{X},p_{Y_{[0:2]}|X})}. \label{eq:1001}
	\end{gather}
	\end{lemma}
	\begin{proof}
	Note that since the probability of $\Psi_{0,j}$ for $j \in \{1,2\}$ is consistent in both \textit{Protocol A} and \textit{Protocol B}, based on \thref{totalvariationlemma} (in Apprendix \ref{AppA}), we consider the proximity conditioned on the event $\Psi_{0,j}$. We make use of the \thref{OSRBnewconditioned} to find criteria, in which those two random distributions would be close in the mean, i.e., for $j \in \{1,2\}$ we can write,
\begin{align}
	\mathbbm{E}&\norm{
		\hat{P}(x^{n},z^{n}_{j},y^{n}_{[0:2]},m^{n}_{[0:2]},f_{[0:2]},\hat{h}_{1},\hat{h}_{2}|\Psi_{0,j})
		-
		P(x^{n},z^{n}_{j},y^{n}_{[0:2]},m^{n}_{[0:2]},f_{[0:2]},\hat{h}_{1},\hat{h}_{2}|\Psi_{0,j})
	}_{TV} 
	\\&  = 
	\mathbbm{E}\norm{
		p^{U}(f_{[0:2]})\phi(x^{n}|\Psi_0)
		-
		P(f_{[0:2]},x^{n}|\Psi_0)
	}_{TV} \label{eq:42}  
	\\& \leq
	2^{-n\aleph(R_{\mathcal{T}},\phi_{X^{n}|\Psi_0},p_{Y_{[0:2]}|X})}, \label{eq:45}
\end{align}
where \eqref{eq:42} follows because other terms in \eqref{eq:44} and \eqref{eq:43} are similar and \eqref{eq:45} follows from \thref{OSRBnewconditioned}. The unconditioned version of this proximity could be stated based on \thref{OSRBnew} as follows.
\begin{align}
	\mathbbm{E}&\norm{
		\hat{P}(x^{n},z^{n}_{j},y^{n}_{[0:2]},m^{n}_{[0:2]},f_{[0:2]},\hat{h}_{1},\hat{h}_{2})
		-
		P(x^{n},z^{n}_{j},y^{n}_{[0:2]},m^{n}_{[0:2]},f_{[0:2]},\hat{h}_{1},\hat{h}_{2}}_{TV} 
	\\&  = 
	\mathbbm{E}\norm{
		p^{U}(f_{[0:2]})\phi(x^{n})
		-
		P(f_{[0:2]},x^{n})
	}_{TV}  
	\\& \leq
	2^{-n\zeta(R_{\mathcal{T}},\phi_{X},p_{Y_{[0:2]}|X})}. \label{eq:55}
\end{align}
\end{proof}

\begin{corollary}
	For right hand sides of \eqref{eq:1000} and \eqref{eq:1001} to converge to zero as $n \rightarrow \infty$, we should have the following conditions met:
	\begin{equation}
		\begin{gathered} 
			\tilde{R}_{0} < H(Y_{0}|X), \\
			\tilde{R}_{1} < H(Y_{1}|X), \\
			\tilde{R}_{2} < H(Y_{2}|X), \\
			\tilde{R}_{0}  + \tilde{R}_{1} < H(Y_{0}Y_{1}|X), \\
			\tilde{R}_{0}  + \tilde{R}_{2} < H(Y_{0}Y_{2}|X), \\
			\tilde{R}_{1}  + \tilde{R}_{2} < H(Y_{1}Y_{2}|X), \\
			\tilde{R}_{0}  + \tilde{R}_{1} + \tilde{R}_{2} <  H(Y_{0}Y_{1}Y_{2}|X), \\ \label{eq:conditionsFirst}
		\end{gathered}
	\end{equation}
\end{corollary}

	\proofpart{3}{Eliminating the shared randomness}
	In this step, we show that the proximity of the main problem's random pmf which is assisted with shared randomness to the random pmf of the dual problem will be preserved if we eliminate the shared randomness by assuming a realization for it. Suppose $(R_{[0:2]},\tilde{R}_{[0:2]})$ satisfy \ref{eq:conditionsSecond} and \ref{eq:conditionsFirst}. \\
		\textit{Type \rom{1} error analysis}: 
		 The type \rom{1} error of the hypothesis testing at Detector $j \in \{1,2\}$ of the main problem assisted with shared randomness can be expressed as: 
		\begin{align}
			\alpha_{n,j} &=
			\hat{P}\condx*{\hat{H}_{j}=1}{H=0} 
			 \\& \leq 
			P\condx*{\hat{H}_{j}=1}{H=0} 
			+
			\mathbbm{E}\norm{\hat{P}-P}_{TV} \label{eq:58}
			\\& \leq
			\epsilon_{n}+\delta_{n}+2^{-n\zeta(R_{\mathcal{T}},q_{X},p_{Y_{[0:2]}|X})} \rightarrow 0, \label{eq:56}
		\end{align}
		where \eqref{eq:58} results from \thref{totalvariationlemma} and \eqref{eq:56} follows from \eqref{eq:57} follows from Lemma \ref{type1dual} and Lemma \ref{type2dual}.
		
		\textit{Type \rom{2} error analysis}:
		By using the same argument as type \rom{1} error, one can find upper bounds for the mean type \rom{2} error probability at Detector $j \in \{1,2\}$.
		\begin{lemma}
			Type \rom{2} error at Detector $j \in \{1,2\}$ of the main problem assisted with shared randomness (\textit{Protocol B}) is bounded as:
			\begin{gather}
				-\frac{1}{n}\lim_{n \rightarrow \infty}\beta_{n,j} \geq
				E_{0,j}(p_{Y_{[0:2]}|X}) \!+\!
				E_{1,j}(p_{Y_{[0:2]}|X}) \!+\!
				E_{2,j}(p_{Y_{[0:2]}|X}).
			\end{gather}
		\end{lemma}
		\begin{proof}
		Recall that the probability of the event $\Psi_{0,j}$ is the same in both problems and depends only on $q_{Z_j}$, therefore we can use the second part of \thref{totalvariationlemma} (in Appendix \ref{AppA}) to obtain: 
		\begin{align}
			\beta_{n,j} &=
			\hat{P}\condx*{\hat{H}_{j}=0}{H=1} 
			\\& \leq
			P\condx*{\hat{H}_{j}=0}{H=1} 
			+
			q(\Psi_{0,j})
			\times
			\mathbbm{E}\norm{\hat{P}\condx*{\cdot}{\Psi_{0,j}}-P\condx*{\cdot}{\Psi_{0,j}}}_{TV}
			 \\& \leq			
			P\condx*{\hat{H}_{j}=0}{H=1} 
			+
			q(\Psi_{0,j})\times2^{-n\aleph(R_{\mathcal{T}},q_{X^{n}|\Psi_{0,j}},p_{Y_{[0:2]}|X})}   \label{eq:84}
			\\& \leq			
			P\condx*{\hat{H}_{j}=0}{H=1} 
			+
			2^{-n\min_{\pi_{Z_j} \in \mathcal{W}^{n}_{j}}\infdiv{\pi_{Z_j}}{q_{Z_j}}}
			\times
			2^{-n\aleph(R_{\mathcal{T}},q_{X^{n}|\Psi_{0,j}},p_{Y_{[0:2]}|X})}   \label{eq:85}
			\\& \leq
			2^{-nE^{n}_{0,j}(p_{Y_{[0:2]}|X})} \!+\!
			2^{-nE^{n}_{1,j}(p_{Y_{[0:2]}|X})} \!+\!
			2^{-n\min_{\pi_{Z_j} \in \mathcal{W}^{n}_{j}}\infdiv{\pi_{Z_j}}{q_{Z_j}}}
			\! \times \!
			2^{-n\aleph(R_{\mathcal{T}},q_{X^{n}|\Psi_{0,j}},p_{Y_{[0:2]}|X})}, \label{eq:62}
		\end{align}
		for 
		\begin{align}
			\mathcal{W}^{n}_{j} \coloneqq
			\left\{
			\pi_{Z_j} \in \mathcal{P}\left(\mathcal{Z}_j\right) \ : \
			\pi_{Z_j}  \stackrel{\delta^{\prime}_{n}}{\approx} \ p_{Z_j}
			\right\}.
		\end{align}
		The inequality in \eqref{eq:84} follows from \thref{OSRBnewconditioned}, \eqref{eq:85} comes from \cite[Problem 2.12]{csiszar11}, and \eqref{eq:62} is deduced by the bound in \eqref{eq:61}.
		The exponent in the last term of \eqref{eq:62} can be further simplified to achieve:
		\begin{align}
			&\min_{\pi_{Z_j} \in \mathcal{W}^{n}_{j}}\infdiv{\pi_{Z_j}}{q_{Z_j}}
			+\aleph(R_{\mathcal{T}},q_{X^{n}|\Psi_{0,j}},p_{Y_{[0:2]}|X})
			=
			\min_{\pi_{Z_j} \in \mathcal{W}^{n}_{j}}\infdiv{\pi_{Z_j}}{q_{Z_j}} 
			\\& +
			\min_{\pi_{Y_{\mathcal{T}},X|Z_j}}\left\{
			\infdiv{\pi_{Y_{\mathcal{T}},X|Z_j}}{p_{Y_{\mathcal{T}},X|Z_j}|p_{Z_j}}
			+
			\frac{1}{2}\left[
			\min_{\mathcal{S} \subseteq \mathcal{T}} \left\{
			H_{\pi}\condx{Y_{\mathcal{S}}}{X} -
			\sum_{i \in \mathcal{S}}R_{i} -
			\delta^{\mathcal{S}}_{n}
			\right\}
			\right]^{+}
			- \epsilon_{n}
			\right\}
			\\& =
			\min_{\pi_{Y_{\mathcal{T}},X,Z_j} \in \mathcal{K}_{2,j}}\left\{
			\infdiv{\pi_{Y_{\mathcal{T}},X,Z_j}}{p_{Y_{\mathcal{T}},X,Z_j}}
			+
			\frac{1}{2}\left[
			\min_{\mathcal{S} \subseteq \mathcal{T}} \left\{
			H_{\pi}\condx{Y_{\mathcal{S}}}{X} -
			\sum_{i \in \mathcal{S}}R_{i} -
			\delta^{\mathcal{S}}_{n}
			\right\}
			\right]^{+}
			- \epsilon_{n}
			\right\},
		\end{align}
	where,
		\[
	\mathcal{K}^{n}_{2,j}= \left\{
	\pi_{X,Y_{[0:2]},Z_j} \in \mathcal{P}\left(\mathcal{X}\times\mathcal{Y}_{[0:2]}\times\mathcal{Z}_j\right) \ : \
	\pi_{Z_j} \stackrel{\delta^{\prime}_{n}}{\approx} p_{Z_j}
	\right\}.
	\]
	For simplicity, we define:
	 \[
	E^{n}_{2,j}(p_{Y_{[0:2]}|X}) \coloneqq 
	\min_{\pi_{Y_{\mathcal{T}},X,Z_j} \in \mathcal{K}_{2,j}}\left\{
	\infdiv{\pi_{Y_{\mathcal{T}},X,Z_j}}{p_{Y_{\mathcal{T}},X,Z_j}}
	+
	\frac{1}{2}\left[
	\min_{\mathcal{S} \subseteq \mathcal{T}} \left\{
	H_{\pi}\condx{Y_{\mathcal{S}}}{X} -
	\sum_{i \in \mathcal{S}}R_{i} -
	\delta^{\mathcal{S}}_{n}
	\right\}
	\right]^{+}
	- \epsilon_{n}
	\right\}.
	\]
		\end{proof}
		Note that the acquired bounds are on random pmfs. Therefore, we can argue that there are fixed binning schemes and $F_{[0:2]}=f_{[0:2]}$ with probability distribution $\tilde{p}$, such that if we replace $P$ with $\tilde{p}$ in \eqref{eq:43}, and name the subsequent distribution with $\hat{p}$, then the type \rom{1} and type \rom{2} error
		probabilities are within a constant multiplicative factor of their mean. \\
		These results are valid if the conditions of \eqref{eq:conditionsFirst} and \eqref{eq:conditionsSecond} are met. The achievable rates using the Fourier-Matzkin elimination algorithm is obtained as,
		\begin{equation}
			\begin{gathered}
				R_{0} > \max_{i \in \{1,2\}}\{ I(X;Y_{0}|Z_{i}) - I(Y_{0},Y_{i}|Z_{i}) \},\\
				R_{1} > I(X;Y_{1}|Z_{1}) - I(Y_{0},Y_{1}|Z_{1}),\\
				R_{2} > I(X;Y_{2}|Z_{2}) - I(Y_{0},Y_{2}|Z_{2}),\\
				R_{0} + R_{1} > I(X;Y_{0}Y_{1}|Z_{1}), \\
				R_{0} + R_{2} >  I(X;Y_{0}Y_{2}|Z_{2}), \\
				R_{0} + R_{1} > I(X;Y_{0}|Z_{2}) + I(X;Y_{1}|Y_{0}Z_{1}) - I(Y_{0};Y_{2}|Z_{2}),\\
				R_{0} + R_{2} > I(X;Y_{0}|Z_{1}) + I(X;Y_{2}|Y_{0}Z_{2}) - I(Y_{0};Y_{1}|Z_{1}),\\
				R_{1} + R_{2} > I(X;Y_{1}|Y_{0}Z_{1}) + I(X;Y_{2}|Y_{0}Z_{2}) + I(Y_{1};Y_{2}|XY_{0}) - I(Y_{1}Y_{2};Y_{0}|X),\\
				R_{0} + R_{1} + R_{2} >   I(X;Y_{1}|Y_{0}Z_{1}) + I(X;Y_{2}|Y_{0}Z_{2}) + \max_{i \in \{1,2\}}\{I(Y_{0};X|Z_{i})\} + I(Y_{1};Y_{2}|XY_{0}),\\
				2R_{0}+R_{1}+R_{2} >I(X;Y_{1}|Y_{0}Z_{1})\!+\!I(X;Y_{2}|Y_{0}Z_{2})\!+\!I(Y_{0};X|Z_{1})\!+\!I(Y_{0};X|Z_{2})\!+\!I(Y_{1};Y_{2}|XY_{0}).
			\end{gathered}
		\end{equation}

	\subsection{Privacy constraints}\label{privacy}
	Now we devise a lower bound on the equivocation measure of the latent random observations, i.e., $S^{n}_{i}$ for $i \in \{1,2\}$. 
	\begin{align}
		H_{\tilde{p}}(S^{n}_{i}|Z^{n}_{i},M_0,M_i) &\geq 
		H_{\tilde{p}}(S^{n}_{i}|Z^{n}_{i},M_0,M_i,Y^{n}_{0},Y^{n}_{i}) 
		\\& = \label{eq:92}
		H_{\tilde{p}}(S^{n}_{i}|Z^{n}_{i},Y^{n}_{0},Y^{n}_{i}) 
		\\& \geq \label{eq:93}
		H_{\hat{p}}(S^{n}_{i}|Z^{n}_{i},Y^{n}_{0},Y^{n}_{i})
		+5\!\times\!2^{-n\zeta(R_{\mathcal{T}},\phi_{X},p_{Y_{[0:2]}|X})}
		\log{\frac{4\!\times\!2^{-n\zeta(R_{\mathcal{T}},\phi_{X},p_{Y_{[0:2]}|X})}}{\abs{\mathcal{S}_{i}}^{n}}}
		\\& =
		H_{\hat{p}}(S^{n}_{i}|Z^{n}_{i},Y^{n}_{0},Y^{n}_{i})
		- o(1)
		\\& = \label{eq:94}
		\sum_{j=1}^{n}
		H_{\hat{p}}(S_{i,j}|Z_{i,j},Y_{0,j},Y_{i,j})
		- o(1)
		\\& =
		nH_{\phi}(S_{i}|Z_{i},Y_{0},Y_{i})
		- o(1)
	\end{align}
	where \eqref{eq:92} follows since $M_0$ and $M_i$ are deterministic functions of $Y^{n}_{0}$ and $Y^{n}_{i}$; \eqref{eq:93} is derived from \thref{entropydistanceconditional}, and \eqref{eq:94} follows since $\hat{p}_{S^{n}_{i},Z^{n}_{i},Y^{n}_{0},Y^{n}_{i}}$ is a product distribution.
	
	\begin{appendices}
		\begingroup
		\allowdisplaybreaks
		\clearpage
		\section{Preliminary Lemmas}
		\label{AppA}
		\begin{lemma}\thlabel{totalproperties}
			Suppose $p_{XY}$ and $q_{XY}$ are two joint probability distributions on $(X, Y)$ with alphabet $\mathcal{X}\times \mathcal{Y}$. Total variation distance has the following properties:
			\begin{enumerate}
				\item{\cite[Property 2]{Sch13}} Let $p_X$ and $q_X$ be marginals of $p_{XY}$ and $q_{XY}$. For $\epsilon \geq 0$ and a bounded function $f(X) \leq b$ where $b \in \mathbbm{R}^{+}$, if $\norm{p_X-q_X}_{TV} \leq \epsilon$, then
				\begin{equation}
					\abs{
						\mathbbm{E}_{p}\left[f(X)\right] - \mathbbm{E}_{q}\left[f(X)\right]
				} \leq \epsilon b.
				\end{equation}
			
			\item{\cite[Lemma 17]{Cover2009CommunicationIN}} Let $p_{X}p_{Y|X}$ and $q_{X}p_{Y|X}$ be two joint distributions on $\mathcal{X}\times\mathcal{Y}$, then 
			\begin{equation}
				\norm{p_{X}p_{Y|X} - q_{X}p_{Y|X}}_{TV} =
				\norm{p_{X}-q_{X}}_{TV}.
			\end{equation}
		
		\item{\cite[Lemma 16]{Cover2009CommunicationIN}} For marginals $p_X$ and $q_X$, the following inequality holds:
		\begin{equation}
			\norm{p_{X}-q_{X}}_{TV} \leq \norm{p_{XY}-q_{XY}}_{TV}.
		\end{equation}
	
		\item{\cite[Lemma 3]{yas14}} If $\norm{p_{X}p_{Y|X}-q_{X}q_{Y|X}}_{TV} \leq \epsilon$, then 
		\begin{equation}
		\mathbbm{E}_{p_{X}}\norm{p_{Y|X}-q_{Y|X}}_{TV} \leq 2\epsilon.
		\end{equation}
		Accordingly, there exists a $x \in \mathcal{X}$ such that $\norm{p_{Y|X=x}-q_{Y|X=x}}_{TV} \leq 2\epsilon$.
			\end{enumerate}
		\end{lemma}

		\begin{lemma} \thlabel{totalvariationlemma}
			Consider two random variables $X$ and $Y$ with two joint probability distributions $p_{XY}$ and $q_{XY}$ on their support set $\mathcal{X} \times \mathcal{Y}$. Given $q(A) \leq \epsilon$ for an arbitrary $A \subseteq \mathcal{X}$, we would have
			\begin{equation}
				p(A) \leq \epsilon + 2\norm{p_X-q_X}_{TV}.
			\end{equation} 
			Also if $p(B)=q(B)$ for some $B \subseteq \mathcal{Y}$, and $q(A \cap B) \leq \delta$ then 
			\begin{equation}
				p(A \cap B) \leq \delta + 2\norm{p_{X|B}-q_{X|B}}_{TV}\times q(B).
			\end{equation}
		\end{lemma}	
		\begin{proof}
			The proof is quite straightforward. One can write the $p(A)$ as follows:
			\begin{align}
				p(A) &= \sum_{x \in A}p(x) \\
				& = \sum_{x \in A}\abs{p(x)} \\
				& = \sum_{x \in A}\abs{p(x) - q(x) + q(x)} \\
				& \leq \sum_{x \in A}\abs{q(x)} + \sum_{x \in A}\abs{p(x) - q(x)} \label{eq:46}\\
				& \leq q(A) + \sum_{x \in \mathcal{X}}\abs{p(x) - q(x)}  \\
				& \leq \epsilon + 2\norm{p-q}_{TV},
			\end{align}
			where \eqref{eq:46} follows from triangle inequality. For the second part we can write
			\begin{align}
				p(A \cap B) &= p(B)p\condx{A}{B} \\
				&=
				p(B)\sum_{x \in A}\abs{p\condx{x}{B}} \\
				&= p(B)\sum_{x \in A}\abs{p\condx{x}{B} - q\condx{x}{B} + q\condx{x}{B}} \\
				&\leq p(B)\sum_{x \in A}\abs{q\condx{x}{B}} + p(B)\sum_{x \in A}\abs{p\condx{x}{B} - q\condx{x}{B}} \\
				&\leq q(B)q\condx{A}{B} + q(B)\sum_{x \in \mathcal{X}}\abs{p\condx{x}{B} - q\condx{x}{B}}  \\
				&\leq \delta + 2\norm{p_{X|B}-q_{X|B}}_{TV}\times q(B).
			\end{align}
		\end{proof}
	
	\begin{lemma}{\cite[Lemma 2.7]{csiszar11}}\thlabel{entropydistance}
		Suppose $p_X$ and $q_X$ are two non-equal pmfs over a discrete random variable $X$ with alphabet $\mathcal{X}$. Given $\Theta \eqqcolon \norm{p_X - q_X}_{TV} \leq \frac{1}{4}$, we have
		\begin{equation}
			\abs{H_{p}(X)-H_{q}(X)} \leq - 2\Theta\log{\frac{2\Theta}{\abs{\mathcal{X}}}}.
		\end{equation}
	\end{lemma}

	\begin{lemma}\thlabel{entropydistanceconditional}
		Let $p_{XY}$ and $q_{XY}$ be two joint distributions on discrete random variables $(X,Y)$ with alphabet $\mathcal{X}\times\mathcal{Y}$. Given $\Theta \eqqcolon \norm{p_{XY} - q_{XY}}_{TV} \leq \frac{1}{2e}$, we have
		\begin{equation}
			\abs{H_{p}\condx{Y}{X}-H_{q}\condx{Y}{X}} 
			\leq -5\Theta \log{\frac{4\Theta}{\abs{\mathcal{Y}}}}.
		\end{equation}
	\end{lemma}
	\begin{proof}
		We begin by using the definition of the conditional entropy as,  
		\begin{align}
			&\abs{H_{p}\condx{Y}{X}-H_{q}\condx{Y}{X}} =
			\abs{
				\sum_{x \in \mathcal{X}}p(x)H_{p}\condx{Y}{X=x} -
				\sum_{x \in \mathcal{X}}q(x)H_{q}\condx{Y}{X=x}
			} 
		\\ &=
		\abs{
			\sum_{x \in \mathcal{X}}p(x)H_{p}\condx{Y\!}{\!X=x} -
			q(x)H_{q}\condx{Y\!}{\!X=x} +
			p(x)H_{q}\condx{Y\!}{\!X=x} -
			p(x)H_{q}\condx{Y\!}{\!X=x}
		}
	\\ &\leq \label{eq:86}
	\abs{
		\sum_{x \in \mathcal{X}}p(x)
		\para*{
		H_{p}\condx{Y\!}{\!X=x} -
		H_{q}\condx{Y\!}{\!X=x}
	}
	} + 
	\abs{
		\sum_{x \in \mathcal{X}}
		\para{ p(x) - q(x) }
		H_{q}\condx{Y\!}{\!X=x}
	}
	\\ &\leq \label{eq:87}
	\abs{
		\sum_{x \in \mathcal{X}}p(x)
		\para*{
			H_{p}\condx{Y\!}{\!X=x} -
			H_{q}\condx{Y\!}{\!X=x}
		}
	} + 
	\Theta\log\abs{\mathcal{Y}}
	\\& \leq \label{eq:88}
	\sum_{x \in \mathcal{X}}p(x)
	\abs{
			H_{p}\condx{Y\!}{\!X=x} -
			H_{q}\condx{Y\!}{\!X=x}
	} + 
	\Theta\log\abs{\mathcal{Y}}
	\\& \leq \label{eq:89}
	\mathbbm{E}_{p_X}\left\{
	-2\Theta_x\log{\frac{2\Theta_x}{\abs{\mathcal{Y}}}}
	\right\}
	 + 
	\Theta\log\abs{\mathcal{Y}},
		\end{align}
	where $\Theta_x \coloneqq \norm{p_{Y|X=x}-q_{Y|X=x}}_{TV}$; \eqref{eq:86} and \eqref{eq:88} follow from the triangle inequality; \eqref{eq:87} follows from part (1) of \thref{totalproperties}; \eqref{eq:89} follows from \thref{entropydistance}. Note that the function $f(x)=-x\log(x)$ is concave, monotonically non-decreasing on the $[0,\frac{1}{e}]$, and $f(x)$ is non-negative on $[0,1]$. Therefore, we can write
	\begin{align}
		\abs{H_{p}\condx{Y}{X}-H_{q}\condx{Y}{X}} &\leq
		\mathbbm{E}_{p_X}\left\{
		-2\Theta_x\log{\frac{2\Theta_x}{\abs{\mathcal{Y}}}}
		\right\}
		+ 
		\Theta\log\abs{\mathcal{Y}}
		\\& =
		\mathbbm{E}_{p_X}\left\{
		-2\Theta_x\log{2\Theta_x}
		\right\}
		+ 
		\mathbbm{E}_{p_X}\left\{
		2\Theta_x\log{\abs{\mathcal{Y}}}
		\right\}
		+
		\Theta\log\abs{\mathcal{Y}}
		\\& \leq \label{eq:90}
		-2\mathbbm{E}_{p_X}\left\{\Theta_x\right\}
		\log{2\mathbbm{E}_{p_X}\left\{\Theta_x\right\}}
		+  
		2\mathbbm{E}_{p_X}\left\{\Theta_x\right\}
		\log{\abs{\mathcal{Y}}}
		+
		\Theta\log\abs{\mathcal{Y}}
		\\& \leq \label{eq:91}
		-4\Theta
		\log{4\Theta}
		+  
		5\Theta\log\abs{\mathcal{Y}}
		\\& \leq
		-5\Theta
		\log{\frac{4\Theta}{\abs{\mathcal{Y}}}},
	\end{align}
	where \eqref{eq:90} follows from the Jensen inequality; The inequality in \eqref{eq:91} is obtained since we have assumed $\Theta=\norm{p_{XY}-q_{XY}}_{TV}$, thereby by using part (4) of the \thref{totalproperties}, we get $\mathbbm{E}_{p_X}\left\lbrace \Theta_x\right\rbrace \leq 2\Theta$; now considering the assumption that $\Theta \leq \frac{1}{2e}$ and the fact that $-2x\log(2x)$ is non-negative and monotonically non-decreasing on $[0, \frac{1}{e}]$, by substituting $\mathbbm{E}_{p_X}\left\lbrace \Theta_x\right\rbrace$ with $2\Theta$, the result follows.
	\end{proof}

	\section{Proof of \thref{OSRBnew}}\label{OSRBnewProof}
		For convenience, let's define $\mathcal{T} \coloneqq [1:T]$ and $M_i \coloneqq 2^{nR_i}$ for $i \in \mathcal{T}$. Recall that the distributed random binning induces the following random pmf on the set $\mathcal{Y}^{n}_{\mathcal{T}}\times\mathcal{X}^{n}\times \prod_{i=1}^{T}[1:M_i]$,
		\begin{equation}
			P(y^{n}_{\mathcal{T}},x^{n},b_{\mathcal{T}})=p(y^{n}_{\mathcal{T}},x^{n})\prod_{i=1}^{T}\mathbbm{1}(\mathcal{B}_i(y^{n}_{i})=b_i). \label{eq:2}
		\end{equation}
		It can be seen that $B_1,B_2,\ldots,B_T$ are uniform and mutually independent of the correlated source $X^{n}$ in the mean, because
		\begin{align}
			\mathbbm{E}_{\mathcal{B}}P(x^{n},b_{\mathcal{T}})&=\mathbbm{E}_{\mathcal{B}}\left\{
			\sum_{y^{n}_{\mathcal{T}} \in \mathcal{Y}^{n}_{\mathcal{T}}}P(y^{n}_{\mathcal{T}},x^{n},b_{\mathcal{T}})
			\right\}              \\
			&=\sum_{y^{n}_{\mathcal{T}} \in \mathcal{Y}^{n}_{\mathcal{T}}}
			p(y^{n}_{\mathcal{T}},x^{n})\mathbbm{E}_{\mathcal{B}}\left\{
			\prod_{i=1}^{T}\mathbbm{1}(\mathcal{B}_i(y^{n}_{i})=b_i)
			\right\} \label{eq:1} \\
			&=\sum_{y^{n}_{\mathcal{T}} \in \mathcal{Y}^{n}_{\mathcal{T}}}
			p(y^{n}_{\mathcal{T}},x^{n})
			\prod_{i=1}^{T}\mathbbm{E}_{\mathcal{B}}\left\{\mathbbm{1}(\mathcal{B}_i(y^{n}_{i})=b_i)
			\right\} \label{eq:3} \\
			&=\sum_{y^{n}_{\mathcal{T}} \in \mathcal{Y}^{n}_{\mathcal{T}}}
			p(y^{n}_{\mathcal{T}},x^{n})
			\prod_{i=1}^{T}\frac{1}{M_i}\label{eq:local4} \\
			&=p(x^{n})
			\prod_{i=1}^{T}\frac{1}{M_i},
		\end{align}
		where \eqref{eq:1} results directly from \eqref{eq:2}, \eqref{eq:3} follows from the independence between each of the random mappings, and \eqref{eq:local4} follows because the random mappings are uniform. From now on, for any $\mathcal{S} \subseteq \mathcal{T}$ we will use the $p^{U}_{\mathcal{S}}=\prod_{i \in \mathcal{S}}\frac{1}{M_i}$ convention. Therefore we have  
		\begin{align}
			\mathbbm{E}_{\mathcal{B}}P(x^{n},b_{\mathcal{T}})=p(x^{n})p^{U}_{\mathcal{T}}. \label{eq:4}
		\end{align}
		We can use \eqref{eq:4} to rephrase the total variation distance between the induced random pmf and its expected value by writing
		\begin{align}
			\mathbbm{E}_{\mathcal{B}}\norm{P(x^{n},b_{\mathcal{T}}) - \mathbbm{E}_{\mathcal{B}}P(x^{n},b_{\mathcal{T}})}_{TV}
			&=
			\mathbbm{E}_{\mathcal{B}}\norm{P(x^{n},b_{\mathcal{T}}) - p(x^{n})p^{U}_{\mathcal{T}})}_{TV} \label{eq:5}                   \\
			&=
			\mathbbm{E}_{\mathcal{B}}\left\{
			\frac{1}{2}
			\sum_{x^{n},b_{\mathcal{T}}}\abs{P(x^{n},b_{\mathcal{T}}) - p(x^{n})p^{U}_{\mathcal{T}}}
			\right\} \label{eq:6}                                                                                           \\
			&=
			\frac{1}{2}
			\sum_{x^{n},b_{\mathcal{T}}}p(x^{n})p^{U}_{\mathcal{T}}\mathbbm{E}_{\mathcal{B}}
			\abs{\frac{P(x^{n},b_{\mathcal{T}})}{p(x^{n})p^{U}_{\mathcal{T}}}-1} \label{eq:local7},
		\end{align}
		where \eqref{eq:6} is due to the very  definition of the total variation distance.

		Now given $(x^{n},b_{\mathcal{T}}) \in \mathcal{X}^{n}\times \prod_{i=1}^{T}[1:M_i]$,  let us define,
		\begin{align}
			L_{B}(x^{n},b_{\mathcal{T}}) & \coloneqq
			\frac{P(x^{n},b_{\mathcal{T}})}{p(x^{n})p^{U}_{\mathcal{T}}} \label{eq:local10}  \\
			& =\frac{1}{p^{U}_{\mathcal{T}}}
			\sum_{y^{n}_{\mathcal{T}} \in \mathcal{Y}^{n}_{\mathcal{T}}}
			p\condx{y^{n}_{\mathcal{T}}}{x^{n}}\prod_{i=1}^{T}\mathbbm{1}(\mathcal{B}_i(y^{n}_{i})=b_i),
			\label{eq:local11}
		\end{align}
		where the definition is confined on the support set of $p(x^n)$. One can observe that $L_{B}(x^{n},b_{\mathcal{T}})$ depends on the random binnings' distribution and therefore it itself is a random variable. It follows from the definition that
		\begin{equation}
			\mathbbm{E}_{\mathcal{B}}\left\{L_{B}(x^{n},b_{\mathcal{T}})\right\}=1. \label{eq:9}
		\end{equation}
		Using \eqref{eq:9} and \eqref{eq:local7}, we can write
		\begin{align}
			\mathbbm{E}_{\mathcal{B}}\norm{P(x^{n},b_{\mathcal{T}}) - \mathbbm{E}_{\mathcal{B}}P(x^{n},b_{\mathcal{T}})}_{TV}
			&=
			\frac{1}{2}
			\sum_{x^{n},b_{\mathcal{T}}}p(x^{n})p^{U}_{\mathcal{T}}\mathbbm{E}_{\mathcal{B}}
			\abs{\frac{P(x^{n},b_{\mathcal{T}})}{p(x^{n})p^{U}_{\mathcal{T}}}-1}                   \\
			&=
			\frac{1}{2}
			\sum_{x^{n},b_{\mathcal{T}}}p(x^{n})p^{U}_{\mathcal{T}}\mathbbm{E}_{\mathcal{B}}\abs{L_{B}(x^{n},b_{\mathcal{T}})-1} \label{eq:7} \\
			&=
			\frac{1}{2}
			\sum_{x^{n},b_{\mathcal{T}}}p(x^{n})p^{U}_{\mathcal{T}}\mathbbm{E}_{\mathcal{B}}\abs{L_{B}(x^{n},b_{\mathcal{T}})-\mathbbm{E}_{\mathcal{B}}\left\{L_{B}(x^{n},b_{\mathcal{T}})\right\}} \label{eq:8},
		\end{align}
		where \eqref{eq:7} stems from \eqref{eq:local10}, and \eqref{eq:8} follows from \eqref{eq:9}. 
		
		Now we use type enumeration method to break down $L_{B}(x^{n},b_{\mathcal{T}})$ into simpler components with more interesting characteristics. Suppose $\pi_{\bar{X}}$ is the type of $x^{n} \in \mathcal{X}^{n}$ and let $\pi_{\bar{Y}_{\mathcal{T}}|\bar{X}}$ denote the conditional type of $y^{n}_{\mathcal{T}} \in \mathcal{Y}^{n}_{\mathcal{T}}$ given $x^{n}$, so that for the joint type $\pi_{\bar{Y}_{\mathcal{T}},\bar{X}}$ of the sequence $(y^{n}_{\mathcal{T}},x^{n})$ we have
		\begin{equation}
			\pi_{\bar{Y}_{\mathcal{T}},\bar{X}}(a_{\mathcal{T}},b)=\pi_{\bar{Y}_{\mathcal{T}}|\bar{X}}(a_{\mathcal{T}}|b)\pi_{\bar{X}}(b),
		\end{equation}
		for every $a_{\mathcal{T}} \in \mathcal{Y}_{\mathcal{T}}$ and $b \in \mathcal{X}$. Note that given $x^{n}$,
		one can partition the elements of $\mathcal{Y}^{n}_{\mathcal{T}}$ in \eqref{eq:local10} into possible conditional types and write,
		\begin{equation}
			L_{B}(x^{n},b_{\mathcal{T}}) = \frac{1}{p^{U}_{\mathcal{T}}}
			\sum_{\pi_{\bar{Y}_{\mathcal{T}}|\bar{X}} \in \mathcal{P}_{n}(\mathcal{Y}_{\mathcal{T}}|\pi_{\bar{X}})}
			N_{\pi_{\bar{Y}_{\mathcal{T}}|\bar{X}}}(x^n,b_{\mathcal{T}})
			l_{\pi_{\bar{Y}_{\mathcal{T}}|\bar{X}}}(x^n),
			\label{eq:local12}
		\end{equation}
		where,
		\begin{equation}
			N_{\pi_{\bar{Y}_{\mathcal{T}}|\bar{X}}}(x^n,b_{\mathcal{T}})\coloneqq\abs{\left\{\
				y^{n}_{\mathcal{T}} : \
				\mathcal{B}_i(y^{n}_{i})=b_i \ for \ i \in \mathcal{T} \ \wedge \
				y^{n}_{\mathcal{T}} \in \mathcal{T}^{n}_{\pi_{\bar{Y}_{\mathcal{T}}|\bar{X}}}(x^n)
				\right\}} \label{eq:11},
		\end{equation}
		is a random variable since it depends on random mappings, i.e., $\mathcal{B}$, and
		\begin{equation}
			l_{\pi_{\bar{Y}_{\mathcal{T}}|\bar{X}}}(x^n) \coloneqq p\condx{y^{n}_{\mathcal{T}}}{x^n},
		\end{equation}
		for some $y^{n}_{\mathcal{T}} \in \mathcal{T}^{n}_{\pi_{\bar{Y}_{\mathcal{T}}|\bar{X}}}(x^n)$. The particular choice of $y^{n}_{\mathcal{T}}$ is irrelevant as long as it provides the specified joint type. Note that $l_{\pi_{\bar{Y}_{\mathcal{T}}|\bar{X}}}(x^n)$ is only dependent on $x^{n}$ through its type. 
		
		Let us define
		\begin{equation}
			Z_{\pi_{\bar{Y}_{\mathcal{T}}|\bar{X}}}(x^n,b_{\mathcal{T}})\coloneqq\ \frac{1}{p^{U}_{\mathcal{T}}}N_{\pi_{\bar{Y}_{\mathcal{T}}|\bar{X}}}(x^n,b_{\mathcal{T}})l_{\pi_{\bar{Y}_{\mathcal{T}}|\bar{X}}}(x^n). \label{eq:81}
		\end{equation}
		From \eqref{eq:local12} and \eqref{eq:81}, we obtain,
		\begin{equation}
			L_{B}(x^{n},b_{\mathcal{T}}) =
			\sum_{\pi_{\bar{Y}_{\mathcal{T}}|\bar{X}}}Z_{\pi_{\bar{Y}_{\mathcal{T}}|\bar{X}}}(x^n,b_{\mathcal{T}}), \label{eq:35}
		\end{equation}
		and thus,
		\begin{align}
			\mathbbm{E}_{\mathcal{B}}\abs{L_{B}(x^{n},b_{\mathcal{T}})-\mathbbm{E}_{\mathcal{B}}\left\{L_{B}(x^{n},b_{\mathcal{T}})\right\}}                                                                                              &=
			\mathbbm{E}_{\mathcal{B}}\abs{\sum_{\pi_{\bar{Y}_{\mathcal{T}}|\bar{X}}}Z_{\pi_{\bar{Y}_{\mathcal{T}}|\bar{X}}}(x^n,b_{\mathcal{T}}) - \sum_{\pi_{\bar{Y}_{\mathcal{T}}|\bar{X}}}\mathbbm{E}_{\mathcal{B}}\left\{Z_{\pi_{\bar{Y}_{\mathcal{T}}|\bar{X}}}(x^n,b_{\mathcal{T}})\right\}} \\
			&\leq
			\sum_{\pi_{\bar{Y}_{\mathcal{T}}|\bar{X}}}
			\mathbbm{E}_{\mathcal{B}}\abs{
				Z_{\pi_{\bar{Y}_{\mathcal{T}}|\bar{X}}}(x^n,b_{\mathcal{T}}) -
				\mathbbm{E}_{\mathcal{B}}\left\{Z_{\pi_{\bar{Y}_{\mathcal{T}}|\bar{X}}}(x^n,b_{\mathcal{T}})\right\}
			} \label{eq:10}
		\end{align}
		where \eqref{eq:10} follows from the triangle inequality. Substituting \eqref{eq:10} into \eqref{eq:8}, we obtain the following bound for our intended distance:
		
		\begin{align}
			\begin{aligned}
			\mathbbm{E}_{\mathcal{B}}\norm{P(x^{n},b_{\mathcal{T}}) - \mathbbm{E}_{\mathcal{B}}P(x^{n},b_{\mathcal{T}})}_{TV}& \\
			\leq
			\frac{1}{2}
			\sum_{x^{n},b_{\mathcal{T}}}p(x^{n})p^{U}_{\mathcal{T}}
			\sum_{\pi_{\bar{Y}_{\mathcal{T}}|\bar{X}}}&
			\mathbbm{E}_{\mathcal{B}}\abs{
				Z_{\pi_{\bar{Y}_{\mathcal{T}}|\bar{X}}}(x^n,b_{\mathcal{T}}) -
				\mathbbm{E}_{\mathcal{B}}\left\{Z_{\pi_{\bar{Y}_{\mathcal{T}}|\bar{X}}}(x^n,b_{\mathcal{T}})\right\}
			}.
		\end{aligned}
		\end{align} 
		Now we would be able to shift our attention from concentration properties of the induced random pmf, namely $P(x^{n},b_{\mathcal{T}})$, to that of $Z_{\pi_{\bar{Y}_{\mathcal{T}}|\bar{X}}}(x^n,b_{\mathcal{T}})$. For this purpose, we are yet to show that one can find upper bounds for the expectation and variance of $N_{\pi_{\bar{Y}_{\mathcal{T}}|\bar{X}}}(x^n,b_{\mathcal{T}})$, independent from $b_{\mathcal{T}}$ and dependent on $x^{n}$ only through its type. Then by using these two upper bounds we can bound the deviations of $Z_{\pi_{\bar{Y}_{\mathcal{T}}|\bar{X}}}(x^n,b_{\mathcal{T}})$ from its mean in two different ways, and thus, $Z_{\pi_{\bar{Y}_{\mathcal{T}}|\bar{X}}}(x^n,b_{\mathcal{T}})$ deviation is less than their minimum. This claim, should it be true, might yield some intuition about why type enumeration method could be a good way to establish upper bounds on $\mathbbm{E}_{\mathcal{B}}\norm{P(x^{n},b_{\mathcal{T}}) - \mathbbm{E}_{\mathcal{B}}P(x^{n},b_{\mathcal{T}})}_{TV}$. We prove this claim and establish two distinct upper bounds on $\mathbbm{E}_{\mathcal{B}}\abs{
			Z_{\pi_{\bar{Y}_{\mathcal{T}}|\bar{X}}}(x^n,b_{\mathcal{T}}) -
			\mathbbm{E}_{\mathcal{B}}\left\{Z_{\pi_{\bar{Y}_{\mathcal{T}}|\bar{X}}}(x^n,b_{\mathcal{T}})\right\}
		}$ before proceeding further.
		Using our definition in \eqref{eq:11}, we have
		\begin{align}
			N_{\pi_{\bar{Y}_{\mathcal{T}}|\bar{X}}}(x^n,b_{\mathcal{T}})
			&=\sum_{y^{n}_{\mathcal{T}} \in \mathcal{Y}^{n}_{\mathcal{T}}}
			\mathbbm{1}\left(
			\mathcal{B}_i(y^{n}_{i})=b_i \ for \ i \in \mathcal{T} \ \wedge \
			y^{n}_{\mathcal{T}} \in \mathcal{T}^{n}_{\pi_{\bar{Y}_{\mathcal{T}}|\bar{X}}}(x^n)
			\right)                                        \\
			&=\sum_{y^{n}_{\mathcal{T}} \in \mathcal{Y}^{n}_{\mathcal{T}}}
			\mathbbm{1}\left(
			y^{n}_{\mathcal{T}} \in \mathcal{T}^{n}_{\pi_{\bar{Y}_{\mathcal{T}}|\bar{X}}}(x^n)
			\right)
			\prod_{i=1}^{T}\mathbbm{1}(\mathcal{B}_i(y^{n}_{i})=b_i). \label{eq:13}
		\end{align}
		Taking expectation results in,
		\begin{align}
			\mathbbm{E}_{\mathcal{B}}\left\{N_{\pi_{\bar{Y}_{\mathcal{T}}|\bar{X}}}(x^n,b_{\mathcal{T}})\right\}
			&=
			\sum_{y^{n}_{\mathcal{T}} \in \mathcal{Y}^{n}_{\mathcal{T}}}
			\mathbbm{E}_{\mathcal{B}}\left\{
			\mathbbm{1}\left(
			y^{n}_{\mathcal{T}} \in \mathcal{T}^{n}_{\pi_{\bar{Y}_{\mathcal{T}}|\bar{X}}}(x^n)
			\right)
			\prod_{i=1}^{T}\mathbbm{1}(\mathcal{B}_i(y^{n}_{i})=b_i)
			\right\} \label{eq:12}                                                                 \\
			&=
			\sum_{y^{n}_{\mathcal{T}} \in \mathcal{Y}^{n}_{\mathcal{T}}}
			\mathbbm{1}\left(
			y^{n}_{\mathcal{T}} \in \mathcal{T}^{n}_{\pi_{\bar{Y}_{\mathcal{T}}|\bar{X}}}(x^n)
			\right)
			\mathbbm{E}_{\mathcal{B}}\left\{
			\prod_{i=1}^{T}\mathbbm{1}(\mathcal{B}_i(y^{n}_{i})=b_i)
			\right\} \label{eq:14}                                                                 \\
			&=
			\sum_{y^{n}_{\mathcal{T}} \in \mathcal{Y}^{n}_{\mathcal{T}}}
			\mathbbm{1}\left(
			y^{n}_{\mathcal{T}} \in \mathcal{T}^{n}_{\pi_{\bar{Y}_{\mathcal{T}}|\bar{X}}}(x^n)
			\right)
			p^{U}_{\mathcal{T}}                                                                          \\
			&=p^{U}_{\mathcal{T}}\abs{\mathcal{T}^{n}_{\pi_{\bar{Y}_{\mathcal{T}}|\bar{X}}}(x^n)} \label{eq:19},
		\end{align}
		where \eqref{eq:12} follows from \eqref{eq:13} and \eqref{eq:14} follows from the fact that the joint type of $(x^{n},y^{n}_{\mathcal{T}})$ is independent from the random mappings. Now if one defines
		\[
		\Gamma_{\pi_{\bar{Y}_{\mathcal{T}}|\bar{X}}}(y^{n}_{\mathcal{T}},x^n,b_{\mathcal{T}}) \coloneqq
		\mathbbm{1}\left(y^{n}_{\mathcal{T}} \in \mathcal{T}^{n}_{\pi_{\bar{Y}_{\mathcal{T}}|\bar{X}}}(x^n)\right)
		\prod_{i=1}^{T}\mathbbm{1}(\mathcal{B}_i(Y^{n}_{i})=b_i),
		\]
		by using \eqref{eq:13} and the identity regarding the variance of sum of random variables, .i.e., $\Var(\sum_{i}X_{i})=\sum_{i,j}\Cov(X_i,X_j)$, we have
		\begin{align}
			\Var&\left(N_{\pi_{\bar{Y}_{\mathcal{T}}|\bar{X}}}(x^n,b_{\mathcal{T}}) \right)
			= \Var\left(
			\sum_{y^{n}_{\mathcal{T}} \in \mathcal{Y}^{n}_{\mathcal{T}}}
			\Gamma_{\pi_{\bar{Y}_{\mathcal{T}}|\bar{X}}}(y^{n}_{\mathcal{T}},x^n,b_{\mathcal{T}})
			\right) \\
			&= \sum_{y^{n}_{\mathcal{T}},\tilde{y}^{n}_{\mathcal{T}} \in \mathcal{Y}^{n}_{\mathcal{T}}}
			\Cov\left(
			\Gamma_{\pi_{\bar{Y}_{\mathcal{T}}|\bar{X}}}(y^{n}_{\mathcal{T}},x^n,b_{\mathcal{T}}),
			\Gamma_{\pi_{\bar{Y}_{\mathcal{T}}|\bar{X}}}(\tilde{y}^{n}_{\mathcal{T}},x^n,b_{\mathcal{T}})
			\right) \\
			&= \sum_{y^{n}_{\mathcal{T}},\tilde{y}^{n}_{\mathcal{T}}}
			\mathbbm{1}\left(y^{n}_{\mathcal{T}},\tilde{y}^{n}_{\mathcal{T}} \in \mathcal{T}^{n}_{\pi_{\bar{Y}_{\mathcal{T}}|\bar{X}}}(x^n)\right)
			\Cov\left(
			\prod_{i=1}^{T}\mathbbm{1}(\mathcal{B}_i(y^{n}_{i})=b_i),
			\prod_{i=1}^{T}\mathbbm{1}(\mathcal{B}_i(\tilde{y}^{n}_{i})=b_i)
			\right).  \label{eq:15}
		\end{align}
		The bin assignment for distinct realizations of $y^{n}_{i} \in \mathcal{Y}^{n}_{i}$ for $i \in \mathcal{T}$ are done independently from each other. Therefore, the covariance terms in \eqref{eq:15} depend only on the subset $\mathcal{S} \subseteq \mathcal{T}$ where we have $y^{n}_{i}=\tilde{y}^{n}_{i}$ for $i \in \mathcal{S}$. It is only natural to partition the set $(y^{n}_{\mathcal{T}},\tilde{y}^{n}_{\mathcal{T}}) \in \mathcal{Y}^{n}_{\mathcal{T}} \times \mathcal{Y}^{n}_{\mathcal{T}}$ into the sets with the same $\mathcal{S}$ where they match,  specified as
		
		\begin{equation}
			\mathcal{K}_{\mathcal{S}} \coloneqq \left\{
			(y^{n}_{\mathcal{T}},\tilde{y}^{n}_{\mathcal{T}}) \in \mathcal{Y}^{n}_{\mathcal{T}} \times \mathcal{Y}^{n}_{\mathcal{T}} \ :
			y^{n}_{\mathcal{S}}=\tilde{y}^{n}_{\mathcal{S}} \ \wedge\ y^{n}_{i} \neq \tilde{y}^{n}_{i}, \ \forall i \in \mathcal{S}^{c}
			\right\}.
		\end{equation}
		Note that for all tuples in $\mathcal{K}_{\emptyset}$, all random mappings are independent, thus the covariance terms are zero. In other words, 
\begin{align}
	\Cov&\left(
			\prod_{i=1}^{T}\mathbbm{1}(\mathcal{B}_i(y^{n}_{i})=b_i),
			\prod_{i=1}^{T}\mathbbm{1}(\mathcal{B}_i(\tilde{y}^{n}_{i})=b_i)
			\right)
			= 0
			\quad \text{for every} \quad 
			(y^{n}_{\mathcal{T}},\tilde{y}^{n}_{\mathcal{T}}) \in \mathcal{K}_{\emptyset}.
\end{align}
For arbitrary $\mathcal{K}_{\mathcal{S}} \neq \mathcal{K}_{\emptyset}$ we can bound the covariance term of $(y^{n}_{\mathcal{T}},\tilde{y}^{n}_{\mathcal{T}}) \in \mathcal{K}_{\mathcal{S}}$ as
		\begin{align}
			\Cov&\left(
			\prod_{i=1}^{T}\mathbbm{1}(\mathcal{B}_i(y^{n}_{i})=b_i),
			\prod_{i=1}^{T}\mathbbm{1}(\mathcal{B}_i(\tilde{y}^{n}_{i})=b_i)
			\right) \leq
			\E_{\mathcal{B}} \left\{
			\prod_{i=1}^{T}\mathbbm{1}(\mathcal{B}_i(y^{n}_{i})=b_i)\mathbbm{1}(\mathcal{B}_i(\tilde{y}^{n}_{i})=b_i)
			\right\}                                                 \label{eq:63}                       
			 \\&\leq
			\E_{\mathcal{B}} \left\{
			\prod_{i \in \mathcal{S}}\mathbbm{1}(\mathcal{B}_i(y^{n}_{i})=b_i)
			\prod_{i \in\mathcal{S}^{c}}\mathbbm{1}(\mathcal{B}_i(y^{n}_{i})=b_i)\mathbbm{1}(\mathcal{B}_i(\tilde{y}^{n}_{i})=b_i)
			\right\}                                               \label{eq:64}                          
			\\& =
			\E_{\mathcal{B}} \left\{
			\prod_{i \in\mathcal{S}}\mathbbm{1}(\mathcal{B}_i(y^{n}_{i})=b_i) \right\}
			\E_{\mathcal{B}}\left\{
			\prod_{i \in\mathcal{S}^{c}}\mathbbm{1}(\mathcal{B}_i(y^{n}_{i})=b_i)
			\prod_{i \in\mathcal{S}^{c}}\mathbbm{1}(\mathcal{B}_i(\tilde{y}^{n}_{i})=b_i)\right\}
			\label{eq:65} 
			\\& =
			\E_{\mathcal{B}} \left\{
			\prod_{i \in\mathcal{S}}\mathbbm{1}(\mathcal{B}_i(y^{n}_{i})=b_i) \right\}
			\E_{\mathcal{B}}\left\{
			\prod_{i \in\mathcal{S}^{c}}\mathbbm{1}(\mathcal{B}_i(y^{n}_{i})=b_i)\right\}
			\E_{\mathcal{B}}\left\{
			\prod_{i \in\mathcal{S}^{c}}\mathbbm{1}(\mathcal{B}_i(\tilde{y}^{n}_{i})=b_i)\right\} 
			\label{eq:66} 
			\\& =
			p^{U}_{\mathcal{S}}\left(p^{U}_{\mathcal{S}^{c}}\right)^{2} \label{eq:16},
		\end{align}
		where \eqref{eq:63} follows since both the random variables $\prod_{i=1}^{T}\mathbbm{1}(\mathcal{B}_i(y^{n}_{i})=b_i)$ and $\prod_{i=1}^{T}\mathbbm{1}(\mathcal{B}_i(\tilde{y}^{n}_{i})=b_i)$ are non-negative, prompting the use of $\Cov\left(X,Y\right)=\E\left\{XY\right\}-\E\{X\}E\{Y\} \leq \E\left\{XY\right\}$ inequality. Also \eqref{eq:64} is valid because we have assumed that $(y^{n}_{\mathcal{T}},\tilde{y}^{n}_{\mathcal{T}}) \in \mathcal{K}_{\mathcal{S}}$, meaning that $y^{n}_{i}=\tilde{y}^{n}_{i}$ for $i \in \mathcal{S}$, and therefore, $\prod_{i \in\mathcal{S}}\mathbbm{1}(\mathcal{B}_i(y^{n}_{i})=b_i)\mathbbm{1}(\mathcal{B}_i(\tilde{y}^{n}_{i})=b_i)=\prod_{i \in \mathcal{S}}\mathbbm{1}(\mathcal{B}_i(y^{n}_{i})=b_i)$. The equation \eqref{eq:65} follows from the fact that random mappings $\mathcal{B}_{i}(\cdot)$ for $i \in \mathcal{S}$ are independent from $\mathcal{B}_{i}(\cdot)$ for $i \in \mathcal{S}^{c}$. Also \eqref{eq:66} follows since $(y^{n}_{\mathcal{T}},\tilde{y}^{n}_{\mathcal{T}}) \in \mathcal{K}_{\mathcal{S}}$ implies that $y^{n}_{i} \neq \tilde{y}^{n}_{i}$ for $i \in \mathcal{S}^{c}$, and thus, the bin assignment $\mathcal{B}_{i}(y^{n}_{i})$ is independent from $\mathcal{B}_{i}(\tilde{y}^{n}_{i})$ for $i \in \mathcal{S}^{c}$. \\
		Subsequently, we can bound the variance by substituting \eqref{eq:16} in \eqref{eq:15} which gives
		\begin{align}
			\Var\left(N_{\pi_{\bar{Y}_{\mathcal{T}}|\bar{X}}}(x^n,b_{\mathcal{T}}) \right) & \\
			\leq
			\sum_{\emptyset \neq \mathcal{S} \subseteq \mathcal{T}}  &
			\sum_{y^{n}_{\mathcal{T}},\tilde{y}^{n}_{\mathcal{T}} \in \mathcal{K}_{\mathcal{S}}}
			\mathbbm{1}\left(y^{n}_{\mathcal{T}},\tilde{y}^{n}_{\mathcal{T}} \in \mathcal{T}^{n}_{\pi_{\bar{Y}_{\mathcal{T}}|\bar{X}}}(x^n)\right)
			p^{U}_{\mathcal{S}}\left(p^{U}_{\mathcal{S}^{c}}\right)^{2}. \label{eq:17}
		\end{align}

		Note that the bound in \eqref{eq:17} is independent of the $b_{\mathcal{T}}$ and depends on $x^{n}$ only through its type. 
		
		Now for every $\pi_{\bar{Y}_{\mathcal{T}}|\bar{X}} \in \mathcal{P}_{n}(\mathcal{Y}_{\mathcal{T}}|\pi_{\bar{X}})$ one can employ the triangle inequality in the form of $\mathbbm{E}\abs{X-\mathbbm{E}X} \leq 2\mathbbm{E}\abs{X}$ to obtain
		\begin{align}
			\mathbbm{E}_{\mathcal{B}}\abs{
				Z_{\pi_{\bar{Y}_{\mathcal{T}}|\bar{X}}}(x^n,b_{\mathcal{T}}) -
				\mathbbm{E}_{\mathcal{B}}\left\{Z_{\pi_{\bar{Y}_{\mathcal{T}}|\bar{X}}}(x^n,b_{\mathcal{T}})\right\}} \label{eq:21} 
			& \leq
			2\mathbbm{E}_{\mathcal{B}}\left\{Z_{\pi_{\bar{Y}_{\mathcal{T}}|\bar{X}}}(x^n,b_{\mathcal{T}})\right\}              
			 \\& =
			\frac{2}{p^{U}_{\mathcal{T}}}l_{\pi_{\bar{Y}_{\mathcal{T}}|\bar{X}}}(x^n)\mathbbm{E}_{\mathcal{B}}\left\{N_{\pi_{\bar{Y}_{\mathcal{T}}|\bar{X}}}(x^n,b_{\mathcal{T}})\right\}, \label{eq:20}
		\end{align}
		for the non-negative random variable $Z_{\pi_{\bar{Y}_{\mathcal{T}}|\bar{X}}}(x^n,b_{\mathcal{T}})$. Substituting \eqref{eq:19} in \eqref{eq:20} we obtain
		\begin{equation}
			\mathbbm{E}_{\mathcal{B}}\abs{
				Z_{\pi_{\bar{Y}_{\mathcal{T}}|\bar{X}}}(x^n,b_{\mathcal{T}}) -
				\mathbbm{E}_{\mathcal{B}}\left\{Z_{\pi_{\bar{Y}_{\mathcal{T}}|\bar{X}}}(x^n,b_{\mathcal{T}})\right\}}
			\leq
			2l_{\pi_{\bar{Y}_{\mathcal{T}}|\bar{X}}}(x^n)\abs{\mathcal{T}^{n}_{\pi_{\bar{Y}_{\mathcal{T}}|\bar{X}}}(x^n)} \label{eq:23}.
		\end{equation}
	
		By use of the Jensen's inequality in the form of $\mathbbm{E}\abs{X-\mathbbm{E}X}=\mathbbm{E}\sqrt{(X-\mathbbm{E}X)^{2}} \leq\sqrt{\mathbbm{E}(X-\mathbbm{E}X)^{2}}$, one can bound the term in \eqref{eq:21} in another way as below
		\begin{align}
			\mathbbm{E}_{\mathcal{B}}&\abs{
				Z_{\pi_{\bar{Y}_{\mathcal{T}}|\bar{X}}}(x^n,b_{\mathcal{T}}) -
				\mathbbm{E}_{\mathcal{B}}\left\{Z_{\pi_{\bar{Y}_{\mathcal{T}}|\bar{X}}}(x^n,b_{\mathcal{T}})\right\}}                                             
			 \leq
			\sqrt{\Var\left(Z_{\pi_{\bar{Y}_{\mathcal{T}}|\bar{X}}}(x^n,b_{\mathcal{T}})\right)}
			\\& =  
			\sqrt{
				\left(
				\frac{l_{\pi_{\bar{Y}_{\mathcal{T}}|\bar{X}}}(x^n)}{p^{U}_{\mathcal{T}}}
				\right)^2
				\Var\left(N_{\pi_{\bar{Y}_{\mathcal{T}}|\bar{X}}}(x^n,b_{\mathcal{T}}) \right)} \label{eq:40}
			\\& \leq
			l_{\pi_{\bar{Y}_{\mathcal{T}}|\bar{X}}}(x^n)
			\sqrt{\sum_{\emptyset \neq \mathcal{S} \subseteq \mathcal{T}}
				\frac{1}{\left(p^{U}_{\mathcal{S}}p^{U}_{\mathcal{S}^{c}}\right)^{2}}
				\sum_{y^{n}_{\mathcal{T}},\tilde{y}^{n}_{\mathcal{T}} \in \mathcal{K}_{\mathcal{S}}}
				\mathbbm{1}\left(y^{n}_{\mathcal{T}},\tilde{y}^{n}_{\mathcal{T}} \in \mathcal{T}^{n}_{\pi_{\bar{Y}_{\mathcal{T}}|\bar{X}}}(x^n)\right)
				p^{U}_{\mathcal{S}}\left(p^{U}_{\mathcal{S}^{c}}\right)^{2}}                                                                     \label{eq:41}    
			\\& =
			l_{\pi_{\bar{Y}_{\mathcal{T}}|\bar{X}}}(x^n)
			\sqrt{\sum_{\emptyset \neq \mathcal{S} \subseteq \mathcal{T}}
				\frac{1}{p^{U}_{\mathcal{S}}}
				\sum_{y^{n}_{\mathcal{T}},\tilde{y}^{n}_{\mathcal{T}} \in \mathcal{K}_{\mathcal{S}}}
				\mathbbm{1}\left(y^{n}_{\mathcal{T}},\tilde{y}^{n}_{\mathcal{T}} \in \mathcal{T}^{n}_{\pi_{\bar{Y}_{\mathcal{T}}|\bar{X}}}(x^n)\right)} \label{eq:22}
		\end{align}
		where \eqref{eq:40} follows from the definition in \eqref{eq:35}, and \eqref{eq:41} obtained by using the results in \eqref{eq:17}. One can write the following upper bound for every $\mathcal{S} \subseteq \mathcal{T}$ and $x^{n} \in \mathcal{T}^{n}_{\pi_{\bar{X}}}$
		\begingroup
		\allowdisplaybreaks
		\begin{align}
			\sum_{y^{n}_{\mathcal{T}},\tilde{y}^{n}_{\mathcal{T}} \in \mathcal{K}_{\mathcal{S}}}&
			\mathbbm{1}\left(y^{n}_{\mathcal{T}},\tilde{y}^{n}_{\mathcal{T}} \in \mathcal{T}^{n}_{\pi_{\bar{Y}_{\mathcal{T}}|\bar{X}}}(x^n)\right)
			\leq
			\sum_{y^{n}_{\mathcal{T}},\tilde{y}^{n}_{\mathcal{T}} \in \mathcal{Y}^{n}_{\mathcal{T}} \times \mathcal{Y}^{n}_{\mathcal{T}}}
			\mathbbm{1}\left(y^{n}_{\mathcal{S}}=\tilde{y}^{n}_{\mathcal{S}}\right)
			\mathbbm{1}\left(y^{n}_{\mathcal{T}},\tilde{y}^{n}_{\mathcal{T}} \in \mathcal{T}^{n}_{\pi_{\bar{Y}_{\mathcal{T}}|\bar{X}}}(x^n)\right) \label{eq:25}                                               
			\\& =
			\sum_{y^{n}_{\mathcal{T}},\tilde{y}^{n}_{\mathcal{T}}}
			\mathbbm{1}\left(y^{n}_{\mathcal{S}}=\tilde{y}^{n}_{\mathcal{S}}\right)
			\mathbbm{1}\left(y^{n}_{\mathcal{T}},\tilde{y}^{n}_{\mathcal{T}} \in \mathcal{T}^{n}_{\pi_{\bar{Y}_{\mathcal{T}}|\bar{X}}}(x^n)\right)
			\mathbbm{1}^{2}\left(x^{n} \in \mathcal{T}^{n}_{\pi_{\bar{X}}}\right) \label{eq:26}                 
			 \\& =
			\sum_{y^{n}_{\mathcal{T}},\tilde{y}^{n}_{\mathcal{T}}}
			\mathbbm{1}\left(y^{n}_{\mathcal{S}}=\tilde{y}^{n}_{\mathcal{S}}\right)
			\mathbbm{1}\left((y^{n}_{\mathcal{T}},x^n) \in \mathcal{T}^{n}_{\pi_{\bar{Y}_{\mathcal{T}},\bar{X}}}\right)
			\mathbbm{1}\left((\tilde{y}^{n}_{\mathcal{T}},x^n) \in \mathcal{T}^{n}_{\pi_{\bar{Y}_{\mathcal{T}},\bar{X}}}\right) \label{eq:27}                                                                  
			\\&
			\begin{aligned}
				=
				\sum_{y^{n}_{\mathcal{T}},\tilde{y}^{n}_{\mathcal{T}} \in \mathcal{T}^{n}_{\pi_{\bar{Y}_{\mathcal{T}}}} \times \mathcal{T}^{n}_{\pi_{\bar{Y}_{\mathcal{T}}}}}
				&
				\mathbbm{1}\left(y^{n}_{\mathcal{S}}=\tilde{y}^{n}_{\mathcal{S}}\right)
				\mathbbm{1}\left((y^{n}_{\mathcal{S}},x^n) \in \mathcal{T}^{n}_{\pi_{\bar{Y}_{\mathcal{S}},\bar{X}}}\right)
				\mathbbm{1}\left(y^{n}_{\mathcal{S}^{c}} \in \mathcal{T}^{n}_{\pi_{\bar{Y}_{\mathcal{S}^{c}}}|\pi_{\bar{Y}_{\mathcal{S}},\bar{X}}}(y^{n}_{\mathcal{S}},x^{n})\right) 
				\\& \times
				\mathbbm{1}\left((\tilde{y}^{n}_{\mathcal{S}},x^n) \in \mathcal{T}^{n}_{\pi_{\bar{Y}_{\mathcal{S}},\bar{X}}}\right)
				\mathbbm{1}\left(\tilde{y}^{n}_{\mathcal{S}^{c}} \in \mathcal{T}^{n}_{\pi_{\bar{Y}_{\mathcal{S}^{c}}}|\pi_{\bar{Y}_{\mathcal{S}},\bar{X}}}(\tilde{y}^{n}_{\mathcal{S}},x^{n})\right)
			\end{aligned} \label{eq:28}                                                                       \\&
			\begin{aligned}
				&=
				\sum_{y^{n}_{\mathcal{S}},\tilde{y}^{n}_{\mathcal{S}} \in \mathcal{T}^{n}_{\pi_{\bar{Y}_{\mathcal{S}}}} \times \mathcal{T}^{n}_{\pi_{\bar{Y}_{\mathcal{S}}}}}
				\mathbbm{1}\left(y^{n}_{\mathcal{S}}=\tilde{y}^{n}_{\mathcal{S}}\right)
				\mathbbm{1}\left((y^{n}_{\mathcal{S}},x^n) \in \mathcal{T}^{n}_{\pi_{\bar{Y}_{\mathcal{S}},\bar{X}}}\right)
				\mathbbm{1}\left((\tilde{y}^{n}_{\mathcal{S}},x^n) \in \mathcal{T}^{n}_{\pi_{\bar{Y}_{\mathcal{S}},\bar{X}}}\right)
				\\& \times
				\sum_{y^{n}_{\mathcal{S}^{c}},\tilde{y}^{n}_{\mathcal{S}^{c}}\in \mathcal{T}^{n}_{\pi_{\bar{Y}_{\mathcal{S}^{c}}}} \times \mathcal{T}^{n}_{\pi_{\bar{Y}_{\mathcal{S}^{c}}}}}
				\mathbbm{1}\left(y^{n}_{\mathcal{S}^{c}} \in \mathcal{T}^{n}_{\pi_{\bar{Y}_{\mathcal{S}^{c}}}|\pi_{\bar{Y}_{\mathcal{S}},\bar{X}}}(y^{n}_{\mathcal{S}},x^{n})\right)
				\mathbbm{1}\left(\tilde{y}^{n}_{\mathcal{S}^{c}} \in \mathcal{T}^{n}_{\pi_{\bar{Y}_{\mathcal{S}^{c}}}|\pi_{\bar{Y}_{\mathcal{S}},\bar{X}}}(\tilde{y}^{n}_{\mathcal{S}},x^{n})\right)
			\end{aligned}
			\\&
			\begin{aligned}
			= &\sum_{y^{n}_{\mathcal{S}} \in \mathcal{T}^{n}_{\pi_{\bar{Y}_{\mathcal{S}}}}} 
			\mathbbm{1}\left((y^{n}_{\mathcal{S}},x^n) \in \mathcal{T}^{n}_{\pi_{\bar{Y}_{\mathcal{S}},\bar{X}}}\right)
			\\& \times
			\sum_{y^{n}_{\mathcal{S}^{c}} \in \mathcal{T}^{n}_{\pi_{\bar{Y}_{\mathcal{S}^{c}}}}}
			\mathbbm{1}\left(y^{n}_{\mathcal{S}^{c}}\in \mathcal{T}^{n}_{\pi_{\bar{Y}_{\mathcal{S}^{c}}}|\pi_{\bar{Y}_{\mathcal{S}},\bar{X}}}(y^{n}_{\mathcal{S}},x^{n})\right)
			\sum_{\tilde{y}^{n}_{\mathcal{S}^{c}} \in \mathcal{T}^{n}_{\pi_{\bar{Y}_{\mathcal{S}^{c}}}}}
			\mathbbm{1}\left(\tilde{y}^{n}_{\mathcal{S}^{c}}\in \mathcal{T}^{n}_{\pi_{\bar{Y}_{\mathcal{S}^{c}}}|\pi_{\bar{Y}_{\mathcal{S}},\bar{X}}}(y^{n}_{\mathcal{S}},x^{n})\right)
			\end{aligned}
			\\& =
			\frac{\abs{\mathcal{T}^{n}_{\pi_{\bar{Y}_{\mathcal{S}},\bar{X}}}}}{\abs{\mathcal{T}^{n}_{\pi_{\bar{X}}}}} \times
			\left(\frac{\abs{\mathcal{T}^{n}_{\pi_{\bar{Y}_{\mathcal{T}},\bar{X}}}}}{\abs{\mathcal{T}^{n}_{\pi_{\bar{Y}_{\mathcal{S}},\bar{X}}}}}\right)^{2}
			\label{eq:29}
			\\& =
			\frac{\abs{\mathcal{T}^{n}_{\pi_{\bar{Y}_{\mathcal{T}},\bar{X}}}}^2}{\abs{\mathcal{T}^{n}_{\pi_{\bar{X}}}}\abs{\mathcal{T}^{n}_{\pi_{\bar{Y}_{\mathcal{S}},\bar{X}}}}}. 
		\end{align}
	\endgroup
		The inequality in \eqref{eq:25} follows from relaxing the constraints in $\mathcal{K}_{\mathcal{S}}$ by waiving the $y^{n}_{i} \neq \tilde{y}^{n}_{i}$ requirement for $i \in \mathcal{S}^{c}$. \eqref{eq:26} follows from the assumption that $x^{n} \in \mathcal{T}^{n}_{\pi_{\bar{X}}}$.  \eqref{eq:27} is the result of $\mathbbm{1}\left((x^n,y^n) \in \mathcal{T}^{n}_{\pi_{\bar{X},\bar{Y}}}\right)=\mathbbm{1}\left(x^n \in \mathcal{T}^{n}_{\pi_{\bar{X}}}\right)\mathbbm{1}\left(y^n \in \mathcal{T}^{n}_{\pi_{\bar{Y}|\bar{X}}}(x^n)\right)$ identity. \eqref{eq:28} is also another application of this identity and the fact that $\mathcal{Y}_{\mathcal{T}}$- and $\mathcal{X}$-marginals of $\pi_{\bar{X},\bar{Y}_{\mathcal{T}}}$ are fixed to be $\pi_{\bar{Y}_{\mathcal{T}}}$ and $\pi_{\bar{X}}$. \eqref{eq:29} follows from \cite[Lemma 15]{yagli19}.
		
		Now that we have two distinct upper bounds for the deviations of $Z_{\pi_{\bar{Y}_{[1:T]}|\bar{X}}}(x^n,b_{\mathcal{T}})$, this term would be less than their minimum. In other words, by combining the results from \eqref{eq:23} and \eqref{eq:22} we attain,
		\begin{align}
			&\mathbbm{E}_{\mathcal{B}}\abs{
				Z_{\pi_{\bar{Y}_{[1:T]}|\bar{X}}}(x^n,b_{\mathcal{T}}) -
				\mathbbm{E}_{\mathcal{B}}\left\{Z_{\pi_{\bar{Y}_{\mathcal{T}}|\bar{X}}}(x^n,b_{\mathcal{T}})\right\}} 
			\\& \leq
			2l_{\pi_{\bar{Y}_{\mathcal{T}}|\bar{X}}}(x^n)\abs{\mathcal{T}^{n}_{\pi_{\bar{Y}_{\mathcal{T}}|\bar{X}}}(x^n)}
			\min\left\{
			1, \
			\frac{1}{2}
			\sqrt{\sum_{\emptyset \neq \mathcal{S}}
				\frac{1}{p^{U}_{\mathcal{S}}} \times
				\frac{\sum_{y^{n}_{\mathcal{T}},\tilde{y}^{n}_{\mathcal{T}} \in \mathcal{K}_{\mathcal{S}}}
					\mathbbm{1}\left(y^{n}_{\mathcal{T}},\tilde{y}^{n}_{\mathcal{T}} \in \mathcal{T}^{n}_{\pi_{\bar{Y}_{\mathcal{T}}|\bar{X}}}(x^n)\right)}
				{\abs{\mathcal{T}^{n}_{\pi_{\bar{Y}_{\mathcal{T}}|\bar{X}}}(x^n)}^{2}}}
			\right\}                                                                                            
			\\& \leq
			2 \! \times \! 2^{-n\left(
				\infdiv{\pi_{\bar{Y}_{\mathcal{T}}|\bar{X}}}{p_{Y_{\mathcal{T}}|X}|\pi_{\bar{X}}}
				\right)}
			\min\left\{
			1, \
			\frac{1}{2}
			\sqrt{\sum_{\emptyset \neq \mathcal{S}}
				\frac{1}{p^{U}_{\mathcal{S}}} \times
				\frac{\sum_{y^{n}_{\mathcal{T}},\tilde{y}^{n}_{\mathcal{T}} \in \mathcal{K}_{\mathcal{S}}}
					\mathbbm{1}\left(y^{n}_{\mathcal{T}},\tilde{y}^{n}_{\mathcal{T}} \in \mathcal{T}^{n}_{\pi_{\bar{Y}_{\mathcal{T}}|\bar{X}}}(x^n)\right)}
				{\abs{\mathcal{T}^{n}_{\pi_{\bar{Y}_{\mathcal{T}}|\bar{X}}}(x^n)}^{2}}}
			\right\}                                                                                            
			\\& \leq
			2 \times 2^{-n\left(
				\infdiv{\pi_{\bar{Y}_{\mathcal{T}}|\bar{X}}}{p_{Y_{\mathcal{T}}|X}|\pi_{\bar{X}}}
				\right)}
			\min\left\{
			1, \
			\frac{1}{2}
			\sqrt{\sum_{\emptyset \neq \mathcal{S}}
				\frac{1}{p^{U}_{\mathcal{S}}} \times
				\frac{1}{\abs{\mathcal{T}^{n}_{\pi_{\bar{Y}_{\mathcal{S}}|\bar{X}}}(x^n)}}}
			\right\} \label{eq:30}                                                                              
			\\& \leq
			2 \times 2^{-n\left(
				\infdiv{\pi_{\bar{Y}_{\mathcal{T}}|\bar{X}}}{p_{Y_{\mathcal{T}}|X}|\pi_{\bar{X}}}
				\right)}
			\min\left\{
			1, \
			\frac{1}{2}
			\sqrt{\sum_{\emptyset \neq \mathcal{S}}
				2^{n\sum_{i \in \mathcal{S}}R_{i}} \times
				2^{-n\left(\cond{\bar{Y}_{\mathcal{S}}}{\bar{X}} - \abs{\mathcal{X}}\abs{\mathcal{Y}_{\mathcal{S}}}
					\frac{\log(n+1)}{n}\right)
			}}
			\right\}, \label{eq:31}
		\end{align}
		where \eqref{eq:30} follows from \eqref{eq:29} and \eqref{eq:31} follows from \cite[Lemma 2.5]{csiszar11} and the definition of $p^{U}_{\mathcal{S}}$. One can bound the minimum term as,
		\begin{align}
			\min&\left\{
			1, \
			\frac{1}{2}
			\sqrt{\sum_{\emptyset \neq \mathcal{S}}
				2^{n\sum_{i \in \mathcal{S}}R_{i}} \times
				2^{-n\left(\cond{\bar{Y}_{\mathcal{S}}}{\bar{X}} - \abs{\mathcal{X}}\abs{\mathcal{Y}_{\mathcal{S}}}
					\frac{\log(n+1)}{n}\right)
			}}
			\right\} 
			\\& \leq 
			\min\left\{
			1, \
			\sqrt{
				2^{T}
				\max_{\emptyset \neq \mathcal{S}}
				2^{n\sum_{i \in \mathcal{S}}R_{i}} \times
				2^{-n\left(\cond{\bar{Y}_{\mathcal{S}}}{\bar{X}} - \abs{\mathcal{X}}\abs{\mathcal{Y}_{\mathcal{S}}}
					\frac{\log(n+1)}{n}\right)
			}}
			\right\} 
			\\& =
			\min\left\{
			1, \
			2^{-\frac{n}{2}\min_{\emptyset \neq \mathcal{S}}\left(\cond{\bar{Y}_{\mathcal{S}}}{\bar{X}} -
				\sum_{i \in \mathcal{S}}R_{i} -
				\abs{\mathcal{X}}\abs{\mathcal{Y}_{\mathcal{S}}} 
				\frac{\log(n+1)}{n}-
				\frac{T}{n}\right)
			}
			\right\} 
			\\& =
			2^{-\frac{n}{2}\left[\min_{\emptyset \neq \mathcal{S}}\left(\cond{\bar{Y}_{\mathcal{S}}}{\bar{X}} -
				\sum_{i \in \mathcal{S}}R_{i} -
				\abs{\mathcal{X}}\abs{\mathcal{Y}_{\mathcal{S}}} 
				\frac{\log(n+1)}{n}-
				\frac{T}{n}\right)\right]^{+}} 
			\\& =
			2^{-\frac{n}{2}\left[\min_{\emptyset \neq \mathcal{S}}\left(\cond{\bar{Y}_{\mathcal{S}}}{\bar{X}} -
				\sum_{i \in \mathcal{S}}R_{i} -
				\delta^{\mathcal{S}}_{n}
				\right)\right]^{+}},
			\label{eq:32}
		\end{align}
		where $\delta^{\mathcal{S}}_{n} \coloneqq \abs{\mathcal{X}}\abs{\mathcal{Y}_{\mathcal{S}}} 
		\frac{\log(n+1)}{n}+
		\frac{T}{n}$ converges to zero as $n \rightarrow \infty$.\\
		\begingroup
		\allowdisplaybreaks
		By combination of \eqref{eq:8}, \eqref{eq:10}, \eqref{eq:31} and \eqref{eq:32} we conclude that 
		\begin{align}
				\mathbbm{E}_{\mathcal{B}}&\norm{P(x^{n},b_{\mathcal{T}}) - \mathbbm{E}_{\mathcal{B}}P(x^{n},b_{\mathcal{T}})}_{TV} 
				\\& \leq
				\sum_{x^{n},b_{\mathcal{T}}}p(x^{n})p^{U}_{\mathcal{T}}
				\sum_{\pi_{\bar{Y}_{\mathcal{T}}|\bar{X}}}
				2^{-n\left(
					\infdiv{\pi_{\bar{Y}_{\mathcal{T}}|\bar{X}}}{p_{Y_{\mathcal{T}}|X}|\pi_{\bar{X}}}
					\right)}
				\times
				2^{-\frac{n}{2}\left[\min_{\emptyset \neq \mathcal{S}}\left(\cond{\bar{Y}_{\mathcal{S}}}{\bar{X}} -
					\sum_{i \in \mathcal{S}}R_{i} -
					\delta_{n}
					\right)\right]^{+}}  \label{eq:67} 
		\\& =
			\sum_{x^{n}}p(x^{n})
			\sum_{\pi_{\bar{Y}_{\mathcal{T}}|\bar{X}}}
			2^{-n\left(
				\infdiv{\pi_{\bar{Y}_{\mathcal{T}}|\bar{X}}}{p_{Y_{\mathcal{T}}|X}|\pi_{\bar{X}}}
				\right)}
			2^{-\frac{n}{2}\left[\min_{\emptyset \neq \mathcal{S}}\left(\cond{\bar{Y}_{\mathcal{S}}}{\bar{X}} -
				\sum_{i \in \mathcal{S}}R_{i} -
				\delta_{n}
				\right)\right]^{+}} \label{eq:33}
			\\&
			\begin{aligned}
				\leq
				\sum_{\pi_{\bar{X}}}2^{-n\infdiv{\pi_{\bar{X}}}{p_{X}}}
				\sum_{\pi_{\bar{Y}_{\mathcal{T}}|\bar{X}}}
				2^{-n\left(
					\infdiv{\pi_{\bar{Y}_{\mathcal{T}}|\bar{X}}}{p_{Y_{\mathcal{T}}|X}|\pi_{\bar{X}}}
					\right)}
				\times
				2^{-\frac{n}{2}\left[\min_{\emptyset \neq \mathcal{S}}\left(\cond{\bar{Y}_{\mathcal{S}}}{\bar{X}} -
					\sum_{i \in \mathcal{S}}R_{i} -
					\delta_{n}
					\right)\right]^{+}} \label{eq:34}
			\end{aligned} 
		\\& =
			\sum_{\pi_{\bar{Y}_{\mathcal{T}},\bar{X}}}
			2^{-n\left(
				\infdiv{\pi_{\bar{Y}_{\mathcal{T}},\bar{X}}}{p_{Y_{\mathcal{T}},X}}
				\right)}
			\times
			2^{-\frac{n}{2}\left[\min_{\emptyset \neq \mathcal{S}}\left(\cond{\bar{Y}_{\mathcal{S}}}{\bar{X}} -
				\sum_{i \in \mathcal{S}}R_{i} -
				\delta_{n}
				\right)\right]^{+}}
			 \\& \leq
			\left(
			n\!+\!1
			\right)^{\abs{\mathcal{X}}\abs{\mathcal{Y}_{\mathcal{T}}}}
			\max_{\pi_{\bar{Y}_{\mathcal{T}},\bar{X}} \in \mathcal{P}_{n}(\mathcal{Y}_{\mathcal{T}}\times\mathcal{X})}\left\{
			2^{-n\left(
				\infdiv{\pi_{\bar{Y}_{\mathcal{T}},\bar{X}}}{p_{Y_{\mathcal{T}},X}}
				-\frac{1}{2}\left[\min_{\emptyset \neq \mathcal{S}}\left(\cond{\bar{Y}_{\mathcal{S}}}{\bar{X}} -
				\sum_{i \in \mathcal{S}}R_{i} -
				\delta_{n}
				\right)\right]^{+}
				\right)} 
			\right\}  \label{eq:82}
			\\& \leq
			\left(
			n \!+\! 1
			\right)^{\abs{\mathcal{X}}\abs{\mathcal{Y}_{\mathcal{T}}}}
			\max_{\pi_{\bar{Y}_{\mathcal{T}},\bar{X}} \in \mathcal{P}(\mathcal{Y}_{\mathcal{T}}\times\mathcal{X})}\left\{
			2^{-n\left(
				\infdiv{\pi_{\bar{Y}_{\mathcal{T}},\bar{X}}}{p_{Y_{\mathcal{T}},X}}
				-\frac{1}{2}\left[\min_{\emptyset \neq \mathcal{S}}\left(\cond{\bar{Y}_{\mathcal{S}}}{\bar{X}} -
				\sum_{i \in \mathcal{S}}R_{i} -
				\delta_{n}
				\right)\right]^{+}
				\right)}
			\right\} 
			\\& =
			2^{-n \times
				\min_{\pi_{\bar{Y}_{\mathcal{T}},\bar{X}} \in \mathcal{P}(\mathcal{Y}_{\mathcal{T}}\times\mathcal{X})}\left\{
				\infdiv{\pi_{\bar{Y}_{\mathcal{T}},\bar{X}}}{p_{Y_{\mathcal{T}},X}}
				-\frac{1}{2}\left[\min_{\emptyset \neq \mathcal{S}}\left(\cond{\bar{Y}_{\mathcal{S}}}{\bar{X}} -
				\sum_{i \in \mathcal{S}}R_{i} -
				\delta_{n}
				\right)\right]^{+}
				- \epsilon_{n}
				\right\}}, \label{eq:83}
		\end{align}
		where $\epsilon_{n} \coloneqq \abs{\mathcal{X}}\abs{\mathcal{Y}_{\mathcal{T}}} 
		\frac{\log(n+1)}{n}$ converges to zero as $n \rightarrow \infty$. \eqref{eq:33} follows because none of the terms in \eqref{eq:67} depends on the specific realization of $b_{\mathcal{T}}$ since all of them are upper bounds that we have obtained in previous parts of the proof. Also \eqref{eq:34} is achieved by partitioning the set of $x^{n} \in \mathcal{X}^{n}$ and using \cite[Lemma 2.6]{csiszar11}.
		\endgroup
		
		\section{Proof of \thref{OSRBnewconditioned}}\label{OSRBnewconditionedProof}
		The proof of \thref{OSRBnewconditioned} is almost identical to the proof of \thref{OSRBnew}, since the steps taken in \eqref{eq:2}-\eqref{eq:33} do not concern themselves with the particular characteristics of $p(x^{n})$, as long as the conditional distribution $p(y^{n}_{[1:T]}|x^{n})$ remains the same. This is indeed the case for \thref{OSRBnewconditioned}, provided that the sources in the problem form a Markov chain, i.e., $Z^{n} \leftrightarrow X^{n} \leftrightarrow Y^{n}_{[1:T]}$, and therefore,
		\[
		p(y^{n}_{[1:T]},x^{n},z^{n})=p(z^{n})p(x^{n}|z^{n})p(y^{n}_{[1:T]}|x^{n}).
		\]
		Knowing that the steps in \eqref{eq:2}-\eqref{eq:33} remain the same, we can proceed by reminding that,
		\begin{align}
			\begin{aligned}
				\mathbbm{E}_{\mathcal{B}}&\norm{P(x^{n},b_{\mathcal{T}}) - \mathbbm{E}_{\mathcal{B}}P(x^{n},b_{\mathcal{T}})}_{TV} 
				\\& \leq
				\frac{1}{2}
				\sum_{x^{n},b_{\mathcal{T}}}p(x^{n})p^{U}_{\mathcal{T}}
				\sum_{\pi_{\bar{Y}_{\mathcal{T}}|\bar{X}} \in \mathcal{P}_{n}(\mathcal{Y}_{\mathcal{T}}|\pi_{\bar{X}})}
				\mathbbm{E}_{\mathcal{B}}\abs{
					Z_{\pi_{\bar{Y}_{\mathcal{T}}|\bar{X}}}(x^n,b_{\mathcal{T}}) 
					\!-\!
					\mathbbm{E}_{\mathcal{B}}\left\{Z_{\pi_{\bar{Y}_{\mathcal{T}}|\bar{X}}}(x^n,b_{\mathcal{T}})\right\}
				},
			\end{aligned}
		\end{align} 
		where $Z_{\pi_{\bar{Y}_{\mathcal{T}}|\bar{X}}}(x^n,b_{\mathcal{T}})$ is defined in \eqref{eq:81}. Using the concentration properties we have acquired for $Z_{\pi_{\bar{Y}_{\mathcal{T}}|\bar{X}}}(x^n,b_{\mathcal{T}})$ in \eqref{eq:31}, we can write,
		\begin{align}
				\mathbbm{E}_{\mathcal{B}}&\norm{P(x^{n},b_{\mathcal{T}}) - \mathbbm{E}_{\mathcal{B}}P(x^{n},b_{\mathcal{T}})}_{TV} 
				\\& \leq
				\sum_{x^{n},b_{\mathcal{T}}}p(x^{n})p^{U}_{\mathcal{T}}
				\sum_{\pi_{\bar{Y}_{\mathcal{T}}|\bar{X}}}
				2^{-n\left(
					\infdiv{\pi_{\bar{Y}_{\mathcal{T}}|\bar{X}}}{p_{Y_{\mathcal{T}}|X}|\pi_{\bar{X}}}
					\right)}
				\times
				2^{-\frac{n}{2}\left[\min_{\emptyset \neq \mathcal{S}}\left(\cond{\bar{Y}_{\mathcal{S}}}{\bar{X}} -
					\sum_{i \in \mathcal{S}}R_{i} -
					\delta_{n}
					\right)\right]^{+}}
			\\& =
			\sum_{x^{n}}p(x^{n})
			\sum_{\pi_{\bar{Y}_{\mathcal{T}}|\bar{X}}}
			2^{-n\left(
				\infdiv{\pi_{\bar{Y}_{\mathcal{T}}|\bar{X}}}{p_{Y_{\mathcal{T}}|X}|\pi_{\bar{X}}}
				\right)}
			2^{-\frac{n}{2}\left[\min_{\emptyset \neq \mathcal{S}}\left(\cond{\bar{Y}_{\mathcal{S}}}{\bar{X}} -
				\sum_{i \in \mathcal{S}}R_{i} -
				\delta_{n}
				\right)\right]^{+}}
			\\& =
			\sum_{x^{n} \in \mathcal{X}^{n}}\sum_{z^{n} \in \mathcal{T}^{n}_{p_{\bar{Z}}}}p(x^{n},z^{n})
			\sum_{\pi_{\bar{Y}_{\mathcal{T}}|\bar{X}}}
			2^{-n\left(
				\infdiv{\pi_{\bar{Y}_{\mathcal{T}}|\bar{X}}}{p_{Y_{\mathcal{T}}|X}|\pi_{\bar{X}}}
				\right)}
			2^{-\frac{n}{2}\left[\min_{\emptyset \neq \mathcal{S}}\left(\cond{\bar{Y}_{\mathcal{S}}}{\bar{X}} -
				\sum_{i \in \mathcal{S}}R_{i} -
				\delta_{n}
				\right)\right]^{+}}
			\\& =
			\sum_{z^{n} \in \mathcal{T}^{n}_{p_{\bar{Z}}}}p(z^{n})
			\sum_{x^{n}}p(x^{n}|z^{n})
			\sum_{\pi_{\bar{Y}_{\mathcal{T}}|\bar{X}}}
			2^{-n\left(
				\infdiv{\pi_{\bar{Y}_{\mathcal{T}}|\bar{X}}}{p_{Y_{\mathcal{T}}|X}|\pi_{\bar{X}}}
				\right)}
			2^{-\frac{n}{2}\left[\min_{\emptyset \neq \mathcal{S}}\left(\cond{\bar{Y}_{\mathcal{S}}}{\bar{X}} -
				\sum_{i \in \mathcal{S}}R_{i} -
				\delta_{n}
				\right)\right]^{+}}
			\\
			&\begin{aligned}
			\leq
			\sum_{z^{n} \in \mathcal{T}^{n}_{p_{\bar{Z}}}}p(z^{n})&
			\sum_{\pi_{\bar{X}|\bar{Z} \in \mathcal{P}_{n}(\mathcal{X}|p_{\bar{Z}})}}
			2^{-n\left(
				\infdiv{\pi_{\bar{X}|\bar{Z}}}{p_{X|Z}|p_{\bar{Z}}}
				\right)} 
			\\& \times
			\sum_{\pi_{\bar{Y}_{\mathcal{T}}|\bar{X}}}
			2^{-n\left(
				\infdiv{\pi_{\bar{Y}_{\mathcal{T}}|\bar{X}}}{p_{Y_{\mathcal{T}}|X}|\pi_{\bar{X}}}
				\right)}
			2^{-\frac{n}{2}\left[\min_{\emptyset \neq \mathcal{S}}\left(\cond{\bar{Y}_{\mathcal{S}}}{\bar{X}} -
				\sum_{i \in \mathcal{S}}R_{i} -
				\delta_{n}
				\right)\right]^{+}}
			\end{aligned}
		\\& =
		\sum_{z^{n} \in \mathcal{T}^{n}_{p_{\bar{Z}}}}p(z^{n})
		\sum_{\pi_{\bar{Y}_{\mathcal{T}},\bar{X}|\bar{Z}}}
		2^{-n\left(
			\infdiv{\pi_{\bar{Y}_{\mathcal{T}},\bar{X}|\bar{Z}}}{p_{Y_{\mathcal{T}},X|Z}|p_{\bar{Z}}}
			\right)}
		2^{-\frac{n}{2}\left[\min_{\emptyset \neq \mathcal{S}}\left(\cond{\bar{Y}_{\mathcal{S}}}{\bar{X}} -
			\sum_{i \in \mathcal{S}}R_{i} -
			\delta_{n}
			\right)\right]^{+}}
		\\& =
		\sum_{\pi_{\bar{Y}_{\mathcal{T}},\bar{X}|\bar{Z}}}
		2^{-n\left(
			\infdiv{\pi_{\bar{Y}_{\mathcal{T}},\bar{X}|\bar{Z}}}{p_{Y_{\mathcal{T}},X|Z}|p_{\bar{Z}}}
			\right)}
		2^{-\frac{n}{2}\left[\min_{\emptyset \neq \mathcal{S}}\left(\cond{\bar{Y}_{\mathcal{S}}}{\bar{X}} -
			\sum_{i \in \mathcal{S}}R_{i} -
			\delta_{n}
			\right)\right]^{+}},
		\end{align}
		where $\epsilon_{n} \coloneqq \abs{\mathcal{X}}\abs{\mathcal{Y}_{\mathcal{T}}} 
		\frac{\log(n+1)}{n}$ goes to zero as $n \rightarrow \infty$. Now by using the same reasoning as \eqref{eq:82}-\eqref{eq:83}, the proof will be concluded.

	\endgroup
	\end{appendices}

	\bibliographystyle{IEEEtranN}
	\clearpage
	\bibliography{References}

\end{document}